\documentclass[12pt]{article}
\usepackage{times}

\usepackage{bm}
\usepackage{caption}
\usepackage{placeins}
\usepackage{amsmath}
\usepackage{amssymb}
\usepackage{natbib}
\usepackage{graphicx}
\usepackage{url}
\usepackage{algorithm2e}
\usepackage{amsthm}

\newtheorem{theorem}{Theorem}
\newtheorem{corollary}{Corollary}

\newtheorem{assumption}{Assumption}
\begin{document}

\setlength{\parskip}{0pt}
\setlength{\parindent}{0pt}
\vspace*{-2cm}

\begin{center}
{\LARGE Learning control variables and instruments\\
for causal analysis in observational data}
\end{center}
\thispagestyle{empty}

\begin{center}
{\large Nicolas Apfel}\\
{\footnotesize University of Innsbruck, Universitaetsstrasse 15, 6020 Innsbruck\\
\texttt{nicolas.apfel@uibk.ac.at}}\\[0.5em]

{\large Julia Hatamyar}\\
{\footnotesize Centre for Health Economics, Heslington, York, YO10 5DD, United Kingdom\\
\texttt{julia.hatamyar@york.ac.uk}}\\[0.5em]

{\large Martin Huber}\\
{\footnotesize University of Fribourg, Bd.\ de P\'{e}rolles 90, 1700 Fribourg, Switzerland\\
\texttt{martin.huber@unifr.ch}}\\[0.5em]

{\large Jannis Kueck}\\
{\footnotesize Heinrich Heine University D\"usseldorf, Universit\"atsstr. 1, 40225 D\"usseldorf, Germany\\
\texttt{kueck@dice.hhu.de}}
\end{center}

\vspace{0.5em}

\begin{abstract}
\small
This study introduces a data-driven, machine learning-based method to detect suitable control variables and instruments for assessing the causal effect of a treatment on an outcome in observational data. Our approach tests the joint existence of instruments, which are associated with the treatment but not directly with the outcome (at least conditional on observables), and suitable control variables, conditional on which the treatment is exogenous, and learns the partition of instruments and control variables from the observed data. The detection of sets of instruments and control variables relies on the condition that proper instruments are conditionally independent of the outcome given the treatment and suitable control variables. We establish the consistency of our method for detecting control variables and instruments under certain regularity conditions, investigate the finite sample performance through a simulation study, and provide an empirical application to health data from the Oregon Health Insurance Experiment.
\end{abstract}

\vspace{0.5em}

\vspace{0.3em}
\noindent\textbf{Acknowledgements:} We have benefited from comments by Niels Richard Hansen and Leonard Henckel.
\thispagestyle{empty}   % <-- removes page number on title page
\pagebreak
\setcounter{page}{1}    % <-- start numbering from next page

\section{Introduction}

Methods for causal analysis, aimed at quantifying the impact of a treatment on an outcome variable, rely on identifying assumptions considered untestable. For example, the well-known selection-on-observables, unconfoundedness, conditional independence, or ignorability assumption requires the treatment to be exogenous when conditioning on observed control variables, hereafter referred to as covariates. %, see for instance the survey by \citet{ImWo08}. 
The selection of covariates is typically justified based on theoretical and/or empirical reasoning, intuition, domain expertise, or prior empirical findings. Nonetheless, in most empirical scenarios, this selection is debatable, given that the optimal set of covariates meeting the selection-on-observables assumption remains fundamentally uncertain.

In this paper, we suggest a machine learning (ML)-based procedure to simultaneously test the presence of (i) covariates satisfying the selection-on-observables (SOO) assumption and (ii) relevant and valid instrumental variables (IVs) in observational data, as well as learning which variables in the data belong to either the set of covariates or IVs. When we refer to relevant and valid IVs, we mean variables that are associated with the treatment (relevance) but have no direct association with the outcome other than through the treatment (validity) conditional on covariates. We demonstrate that appropriate sets of covariates satisfying the identification requirements for treatment effects based on the SOO assumption, as well as relevant and valid IVs, can be detected in a data-driven way instead of being assumed by the researcher. For testing and learning covariates and instruments, we exploit a conditional independence condition that must hold when both relevant and valid instruments as well as covariates that satisfy the SOO assumption exist: The IVs must be conditionally independent of the outcome, given the treatment and the covariates, see for instance the discussions in  \citet{deLunaJohansson2012}, \citet{BlackJooLaLondeSmithTaylor2015}, and \citet[HK,][]{huberkueck2022}.

The contributions of this paper are two-fold. 
(1) We propose the first data-driven method that is able to assess identification without having to impose SOO or IV validity a priori and which can learn the partition of variables into controls and IVs. We consider a setting with nonlinearities and without assuming homogeneous treatment effects. We also show the theoretical, large-sample properties of our procedure. (2) We propose a new orthogonalized score which can be interesting in more general settings and which is also normally distributed under the alternative hypothesis. 
A detailed literature review can be found in appendix \ref{app:litrev}. 

Our approach consists of the following steps. First, within the combined set of potential covariates and IVs, we sequentially test which variable is strongly associated with the treatment conditional on all remaining variables in that set. Second, we consider each of these strong predictors of the treatment as candidate IVs and sequentially test whether each of them is conditionally independent of the outcome when controlling for the treatment and all remaining variables in the combined set of potential covariates and IVs. If the conditional independence assumption is satisfied by (at least) one candidate IV, then the instrument validity and SOO assumptions hold. This implies that the treatment is as good as random conditional on the remaining variables within the combined set of potential covariates and IVs. Treatment effects can then be estimated using methods that control for observed covariates, such as matching, regression, inverse probability weighting, or doubly robust techniques \citep[][]{huber2023causal}. 
In other words, the output of the algorithm is a decision on whether the SOO assumption is fulfilled and whether the researcher should continue with their analysis. Indeed, what we require is that the algorithm reliably indicates the presence of at least one valid IV while the SOO assumption holds; we do not require, nor can we guarantee, the correct selection of all valid IVs. 

Our test focuses on the conditional mean (rather than full) independence of the IV, which implies the identification of average treatment effects (ATE). When assuming a limited set of observed variables (relative to the sample size), we employ regression for both selecting the candidate IVs in the first step and testing the conditional mean independence of the IV in the second step. More concisely, testing is based on the mean squared difference in outcome prediction when regressing the outcome (1) on the treatment, the control variables, and the candidate IV and (2) on the treatment and the control variables (but not the candidate IV). This approach builds on the mean squared difference test based on a quadratic score function in HK, but applies it sequentially across all candidate IVs. Our test flips the original setup: the $H_0$ states that identification fails; the $H_1$ states it holds. We demonstrate that our method is consistent for correctly determining identification, which we illustrate in a simulation study with ten covariates. As a word of caution for empirical applications, we find that the test might require a large sample. %rather with more than just several 1000 observations.
% MENTION REGULARITY CONDITIONS HERE. In a simulation study with ?? covariates, we find that the test to perform decently in terms of empirical size and power even under sample sizes of ?? or ?? observations. Furthermore, we extend our test to the high dimensional context of a large set of possible covariates and instruments and apply the double machine learning (DML) framework of \cite*{doubleML} to control for covariates in a data-driven manner. 

We apply our method to health data from the Oregon Health Insurance (OHI) Experiment, previously analysed by \citet{finkelstein2012oregon}, in which low-income adults were randomly assigned the opportunity to apply for Medicaid, a public health insurance program in the US. The random assignment provides a plausible IV for actual Medicaid enrollment - the treatment of interest - provided that assignment itself has no direct effect on health outcomes such as doctor visits. Our approach indeed selects random assignment as a valid IV, %together with several other variables, 
and indicates that Medicaid enrollment is exogenous, conditional on a rich set of more than 200 pre-assignment covariates.% capturing a wide range of socio-economic and health-related characteristics.

The remainder of this study is organized as follows. Section \ref{Assumptions} discusses the identifying assumptions. Based on these, Section \ref{testapproach} proposes ML-based procedures for jointly testing the IV and SOO assumptions. Section \ref{selmethod} suggests an algorithm that detects strong and valid IVs as well as covariate sets satisfying the SOO assumption. 
%Sections \ref{testapproach} and \ref{testapproach2} suggest different testing approaches depending on whether the supposed instruments and covariates are predefined by the researcher or to be detected in a data-driven way.
Section \ref{simulation} provides a simulation study analyzing the finite sample performance of our method. Section \ref{application} presents an application to the Oregon Health Insurance Experiment. Section \ref{co} concludes. All proofs, an extension to the multivalued IV case, a literature review, pseudo-code and the full simulation results can be found in the appendix.

\section{Identifying assumptions and testable conditional independence}\label{Assumptions}

First, we briefly review the implication that we will use for testing. \citet{deLunaJohansson2012}, \citet{BlackJooLaLondeSmithTaylor2015}, and HK imply that under IV validity and a SOO assumption concerning the treatment, the IV is conditionally independent of the outcome given the treatment and observed covariates. 
To formalize the assumptions, let us denote by $D$ a treatment whose causal effect on an outcome variable $Y$ is of interest. Both $D$ and $Y$ might be discretely or continuously distributed. Using the potential outcomes framework \citep{Neyman23, Rubin74}, we denote by $Y(d)$ the potential outcome when exogenously setting the treatment $D$ of a subject to value $d$ in the support of the treatment. More generally, we will use capital and lower case letters for referring to random variables and specific values thereof, respectively.\footnote{By representing the potential outcome $Y(d)$ as a function solely dependent on a subject's own treatment status $D=d$, we implicitly adhere to the assumption that the potential outcomes of one subject are not influenced by the treatment status of others. This is known as the `stable unit treatment value assumption' \citep[SUTVA, see the discussion in][]{Rubin80, Cox58}, and is invoked throughout.} Furthermore, we denote by $X$ and $Z$  sets of observed covariates and IVs, whose properties are yet to be defined. Based on this notation, we consider the same identifying assumptions as HK.

The first assumption imposes some causal structure. It rules out the existence of reverse causality\footnote{This means that the outcome cannot causally influence any other variables, and the treatment cannot causally affect any variables other than the outcome. This is in line with the conventional practice of measuring covariates and IVs before treatment assignment, eliminating the potential for reverse causality between $D$ and $Y$ and the pre-treatment variables $X$ and $Z$.} and enforces the principle of causal faithfulness. %\footnote{Causal faithfulness dictates that only variables which are d-separated in the sense of \citet{Pearl1988Probabilistic}, i.e.\ not associated with each other via some causal paths (possibly conditional on other variables) are statistically independent (or conditionally independent). While d-separation is generally a sufficient condition for the (conditional) independence of two variables, it is a necessary condition under causal faithfulness.}. 
We formalise this causal structure using the previously mentioned potential outcome notation, by applying the latter also to other variables. To this end, let \(A(b)\) and \(A(b,c)\) correspond to the potential value of variable \(A\) when setting variable \(B\) to \(b\), or variables \(B\) and \(C\) to \(b\) and \(c\), respectively.
\begin{assumption}[Causal structure]\label{ass0}
\begin{eqnarray*}
D(y)=D,\quad X(d,y)=X,\textrm{ and }Z(d,y)=Z\quad  \forall d \in \mathcal{D}\textrm{ and }y  \in \mathcal{Y},
\end{eqnarray*}
only variables which are d-separated in some causal model are statistically independent.
\end{assumption}
\noindent $\mathcal{D}$ and $\mathcal{Y}$ denote the support of $D$ and $Y$, respectively. The first line of A \ref{ass0} rules out a causal effect of outcome $Y$ on $D$, $X$, or $Z$ and of treatment $D$ on $X$ or $Z$. However, it allows for the possibility of both $X$ and $Z$ affecting $D$, $Y$, or even each other. This assumption aligns with the directed acyclic graph \citep[DAG, see e.g.][]{Pearl00} presented in Figure \ref{causalchain}, where causal relationships between variables are indicated by arrows:  $Z$ and $X$ affect $D$, and $D$ and $X$ affect $Y$. Additionally, $X$ may influence $Z$ or vice versa, denoted by the bidirectional arrow. The DAG also features unobserved terms $U$ and $V$ that affect $Y$ and $D$, respectively, with dashed arrows denoting the unobservable nature of these effects. The second line of A \ref{ass0} enforces causal faithfulness, ensuring that only variables which are d-separated in the sense of \citet{Pearl1988Probabilistic}, i.e.\ not associated with each other via some causal paths (possibly conditional on other variables) are statistically independent (or conditionally independent).\footnote{
d-separation relies on blocking causal paths between variables. Formally, a path between two (sets of) variables $A$ and $B$ is blocked when conditioning on a (set of) control variable(s) $C$ if 
\begin{enumerate}
	\item the path between $A$ and $B$ is a causal chain, implying that $A\rightarrow M \rightarrow B$ or $A\leftarrow M \leftarrow B$, or a confounding association, implying that $A\leftarrow M \rightarrow B$, and variable (set) $M$ is among control variables $C$ (i.e.\ controlled for),
	\item the path between $A$ and $B$ contains a collider, implying that $A\rightarrow  S  \leftarrow B$, and variable (set) $S$ or any variable (set) causally affected by $S$ is not among control variables $C$ (i.e.\ not controlled for).
\end{enumerate}
Based on this definition of blocking, the d-separation criterion states that $A$ and $B$ are d-separated when conditioning on control variable(s) $C$ if and only if $C$ blocks all paths between $D$ and $Y$.} While d-separation is generally a sufficient condition for the (conditional) independence of two variables, it is a necessary condition under causal faithfulness.

%\tikzstyle{EdgeStyle}   = [->,>=stealth']
\begin{figure}[!htp]
	\begin{center}
	    \caption{\label{causalchain}  Causal graph satisfying Assumption \ref{ass0}}\bigskip
    \includegraphics[width=1\textwidth]{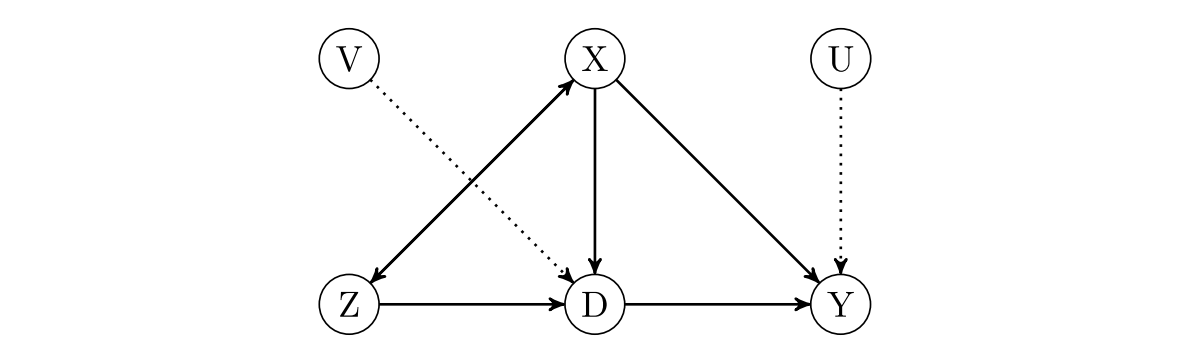}	
\end{center} 
\begin{scriptsize}
    \textit{Note: Z is an instrument, D is the treatment variable, Y is the outcome, X is an observed control, V and U are unobservables, with dotted lines denoting that the relation between variables is unobservable.}
\end{scriptsize}
 %\begin{tikzpicture}[scale=1.3]
%	\SetGraphUnit{2}
%	\Vertex{Z}  \EA(Z){D}  \EA(D){Y} \NO(D){X}  \NO(Y){U}  \NO(Z){V}
 %   \Edges(Z,D,Y) \Edges(X,D) \Edges(X,Y)  \Edges(X,Z) \Edges(Z,X) \Edges[style={dotted}](V,D) \Edges[style={dotted}](U,Y)
%	\end{tikzpicture}
\end{figure}
\noindent The second assumption is a common support assumption concerning the treatment and the IV:
\begin{assumption}[Common support]\label{asscommonsupport}
\begin{eqnarray*}
\mathbb{P}(D=d, Z=z|X)>0 \quad \forall d \in \mathcal{D}\textrm{ and } z \in \mathcal{Z},
\end{eqnarray*}
\end{assumption}
\noindent where $\mathcal{Z}$ denotes the support of $Z$. %For discretely distributed treatments and IVs, A \ref{asscommonsupport} states that the joint probabilities of any $D=d$ and $Z=z$ conditional on $X$ are larger than zero. 
Under continuous treatment and/or IV variables, joint probabilities are to be replaced by joint density functions conditional on $X$. A \ref{asscommonsupport} implies that both $\mathbb{P}(D=d|X)$, the so-called treatment propensity score, and $\mathbb{P}(Z=z|D,X)$, the IV propensity score, are larger than zero. %\footnote{A \ref{asscommonsupport} is strictly speaking not required for testing, if testing only proceeds in a subgroup satisfying common support. However, this may imply reduced testing power.} 
The third assumption imposes a statistical association between the IV and the treatment conditional on the covariates and works as a relevance or first stage assumption.
\begin{assumption}[Conditional dependence between the treatment and instrument]\label{ass3}
\begin{eqnarray*}
D \not\!\perp\!\!\!\perp Z|X,
\end{eqnarray*}
\end{assumption}
\noindent where $\not\!\perp\!\!\!\perp$ denotes statistical dependence. %Together with A \ref{ass0}, which rules out effects of $D$ on $Z$, A \ref{ass3} either implies that $Z$ causally affects $D$, which is known as first stage effect in the IV literature,  or that some (unobserved) characteristics jointly affect $Z$ and $D$ given $X$. 
This assumption is satisfied in Figure \ref{causalchain}, where $Z$ affects $D$. 
The fourth assumption invokes SOO, i.e.\ quasi-random treatment assignment conditional on $X$ as e.g.\ considered in \citet{Im04}:
\begin{assumption}[Conditional independence of the treatment]\label{ass1}
\begin{eqnarray*}
Y(d) {\perp\!\!\!\perp} D | X\quad \forall d \in \mathcal{D},
\end{eqnarray*}
\end{assumption}
\noindent where ${\perp\!\!\!\perp}$ denotes statistical independence. A \ref{ass1} implies that conditional on covariates $X$, there exist no unobserved confounders jointly affecting outcome $Y$ and treatment $D$. 
%The following invokes that the IV is conditionally independent of the potential outcomes given $X$ (validity).
\begin{assumption}[Conditional independence of the instrument]\label{ass2}
\begin{eqnarray*}
Y(d) {\perp\!\!\!\perp} Z | X\quad \forall d \in \mathcal{D}.
\end{eqnarray*}
\end{assumption}
\noindent A \ref{ass2} rules out unobserved confounders jointly affecting $Y$ and $Z$ when controlling for $X$. %This is analogous to A \ref{ass1}, but now concerning $Z$ rather than $D$. 
%Moreover, \ref{ass2} implies an exclusion restriction, i.e. that the IV does not have a direct effect, which is not mediated by $D$, on the potential outcome because the latter is only a function of $d$, conditional on $X$. %(and not $z$), 
Moreover, by assuming that the potential outcome is solely a function of $d$ (and not $z$), A \ref{ass2} also implies that the IV does not have a direct impact on the outcome, except through its impact on the treatment, conditional on X. By this exclusion restriction, it holds that conditional on $X$, $Y(d,z)=Y(d,z')=Y(d)$ for any IV values $z$ and $z'$.
Otherwise, A\ref{ass2} would be violated, because it would follow that $Y(d)=Y(d,Z)$ and $Y(d,Z) \not\!\perp\!\!\!\perp Z | X$. 
A\ref{ass0} is maintained and defines the causal ordering of $D$, $Y$ and candidate variables. Propensity score trimming ensures that A\ref{asscommonsupport} holds. A\ref{ass3} is evaluated in the first-stage screening. A\ref{ass1} and A\ref{ass2} are not imposed, instead, the procedure tests a joint implication of these assumptions via conditional mean-independence. 
When this restriction is supported for at least one candidate IV, this provides empirical support for identification of the ATE via covariate adjustment, without requiring prior knowledge of which variables are valid IV or controls.
%It is worth noting that A \ref{ass2} alone is not sufficient for identifying causal effects based on the instrument, like the local average treatment effect (LATE) on the subpopulation whose treatment reacts to (or complies with) the instrument, see \cite*{Imbens+94} and \cite*{Angrist+96}. In fact, the IV-based assessment of the LATE hinges on further assumptions, like e.g.\ the (conditional) monotonicity of $D$ in $Z$, which we do not consider here, as it is not required for the proposed identification test.

\iffalse
\begin{assumption}\label{ass:nocollider}No colliders\\
	$X(u) = X$
\end{assumption}

\begin{figure}\label{fig:nocollider}
	\caption{No colliders}
\includegraphics[scale=1.6]{DAG_Collider_2.pdf}\includegraphics[scale=1.6]{DAG_Collider_1.pdf}
\begin{footnotesize}
\textit{Note: The left DAG shows a setting where the causal relationship goes from $U$ to $X$. Controlling for $X$ introduces dependence between $Z$ and $Y$, because $X$ now acts as a collider. If the causal relationship goes from $X$ to $U$, as shown in the right graph, this is not the case.} 
\end{footnotesize}
\end{figure}
\fi

These assumptions can be used to test for the identification of causal effects. HK's Theorem 1 demonstrates that conditional on A \ref{ass0} and \ref{ass3}, $Y {\perp\!\!\!\perp}  Z | D=d, X$, the testable conditional independence, is necessary and sufficient for the joint satisfaction of A \ref{ass1} and \ref{ass2} when considering potential outcomes $Y(d)$ which match the factual treatment assignment $D=d$.
Formally,
\begin{eqnarray}\label{mainresult}
Y(d) {\perp\!\!\!\perp}  D |  X,\quad Y(d) {\perp\!\!\!\perp}  Z |  X \iff  Y {\perp\!\!\!\perp}  Z | D=d, X \quad\forall d \in \mathcal{D}.
\end{eqnarray}

Instead of verifying $Y {\perp\!\!\!\perp}  Z | D, X$, HK test conditional mean independence of the IV:
 \begin{eqnarray}\label{mainresult2}
 \mathbb{E}[Y|D,X]=\mathbb{E}[Y|D,X,Z].
 \end{eqnarray}
 Condition \eqref{mainresult2} is sufficient when considering the identification of average causal effects such as the conditional average treatment effect (CATE) given $X$, $\mathbb{E}[Y(1)-Y(0)|X]$, or the average treatment effect (ATE), $\mathbb{E}[Y(1)-Y(0)]$.%, rather than distributional parameters like quantile treatment effects. 
 Theorem 2 in HK shows that \eqref{mainresult2} holds when replacing  Assumptions \ref{ass1} and \ref{ass2} by the weaker conditional mean independence assumptions $\mathbb{E}[Y(d)| D, X]=\mathbb{E}[Y(d)|  X]$ and $\mathbb{E}[Y(d)| Z, X]=\mathbb{E}[Y(d)|  X]$ $\forall d \in \mathcal{D}$, as well as A \ref{ass3} by the first stage condition $\mathbb{E}[D|Z,X] \neq \mathbb{E}[D|X]$, implying that the conditional mean of $D$ varies with $Z$.
 Formally, conditional on A \ref{ass0} and $\mathbb{E}[D|X,Z] \neq \mathbb{E}[D|X]$, it holds that
\begin{eqnarray}\label{mainresult3}
&&\mathbb{E}[Y(d)| D, X]=\mathbb{E}[Y(d)|  X],\quad \mathbb{E}[Y(d)|X, Z]=\mathbb{E}[Y(d)|  X]\\ &\iff&  \mathbb{E}[Y|D=d,X, Z]=\mathbb{E}[Y|D=d,X] \quad\forall d \in \mathcal{D}.\notag
\end{eqnarray}
The testable implication $\mathbb{E}[Y|D=d,X, Z]=\mathbb{E}[Y|D=d,X]$ is necessary and sufficient for the joint satisfaction of conditional mean independence of the treatment and the IV when considering potential outcomes $Y(d)$ matching the factual treatment assignment $D=d$. This in turn is sufficient for identifying the conditional mean
potential outcome $E[Y(d)|X]$ based on observations with $D=d$ when controlling for $X$ and thus, for identifying the CATE and ATE, see for instance the discussion of weak unconfoundedness in \citet{Im00}. This matches the single world intervention graphs (SWIG) framework of \citet{richardson2013single} under the Finest Fully Randomized Causally Interpretable Structured Tree Graph (FFRCISTG) semantics of \citet{Ro86}.

The testing approach of HK requires the prior specification of the IV, $Z$, and covariates, $X$. In contrast, our testing approach, as introduced below, does not require the predefinition of $Z$ and $X$ when testing \eqref{mainresult2}. Instead, it learns them from the data by iteratively considering variables as IV $Z$. This feature appears attractive in many practical contexts where obvious IVs are not available.

\section{Testing based on double machine learning}\label{testapproach}
We henceforth suggest a testing approach based on DML
%, and to this end consider vectors $(Y,Z,D,X)$ such that
%\begin{align}
%Y&=\mu(Z,D,X)+\varepsilon,\quad \mathbb{E}[\varepsilon|X,D,Z]=0,\label{DR_model_1}
% Z&=p(D,X)+\nu, \quad\mathbb{E}[\nu|D,X]=0,\label{DR_model_2}
%\end{align}
%with $\mu(d,x,z):=\mathbb{E}[Y|D=d,X=x, Z=z]$
% and $p(d,x)=\Pr(Z=1|D=d,X=x)$ denoting the conditional instrument probability or instrument propensity score 
%denoting the conditional expectation and $\varepsilon$ the deviation from the conditional mean of $Y$. 
% Further, define $m(d,x):=\mathbb{E}[Y|D=d,X=x]$ and
% $\tilde{\varepsilon}=Y-m(D,X)$. 
%Testing is 
based on the following null hypothesis $H_0$, which is equivalent to the conditional mean independence of the IV provided in condition \eqref{mainresult2}:
\begin{eqnarray}\label{nullhyp1}
H_0: \mathbb{E}[Y|D=d,X=x, Z=z]-E[Y|D=d,X=x]=0\quad\forall d \in  \mathcal{D}, x \in \mathcal{X}, z\ \in \mathcal{Z}.
\end{eqnarray}
Under the null, $H_0$, the mean conditional outcome is constant across values of $Z$ given any value of $D$ and $X$, which may be tested for any values of $D$, $X$, and $Z$ in their respective support. However, if one or several variables are of rich support, this implies many testable implications. %- even infinitely many if one or more variables are continuous. 
%We briefly discuss the possibility, Shorten the following if necessary
For this reason, one possible testing approach is to follow HK and test violations of \eqref{nullhyp1} based on the mean squared difference between the conditional mean outcome when including versus excluding the IV in the conditioning set. Denoting the conditional means by $\mu(d,x,z)=\mathbb{E}[Y|D=d,X=x, Z=z]$ and $m(d,x)=\mathbb{E}[Y|D=d,X=x]$, one aims at testing the following implication of eq. \eqref{nullhyp1}: 
\begin{align}\label{H0_quadratic}
    \mathbb{E}[\left(\mu(D,X,Z)-m(D,X)\right)^2]=0,
\end{align}
based on a moment condition which uses the following \citet{Neyman1959}-orthogonal score:
\begin{align}\label{score2}
\phi(W,\theta,\eta)=(\eta_1(W)-\eta_2(W))^2-\theta+\zeta.
\end{align}
$W=(Y,D,X,Z,\zeta)$ are random variables and $\eta=(\eta_1,\eta_2)$ are the so-called nuisance parameters, whose true values correspond to  $\eta_{0,1}(W)=\mu(D,X,Z)$ and $\eta_{0,2}(W)=m(D,X)$.
We note that the independent mean-zero random variable $\zeta$ in \eqref{score2}
% \begin{align}\label{score3}
% \phi(W,\theta,\eta)=(\eta_1(W)-\eta_2(W))^2-\theta+\zeta,
% \end{align}
is added to avoid a degenerate distribution of the estimator under $H_0$, a common problem in specification tests, see e.g. \citet{hong1995consistent} and \citet{wooldridge1992test}. A disadvantage of testing based on the score function in eq. \eqref{score2} is the requirement to choose a random term $\zeta$, as the optimal selection of  $\zeta$ in a given dataset is generally unknown. Further, while the estimator based on eq. \eqref{score2} is asymptotically normal under the null hypothesis, as demonstrated in HK, this is generally not the case under the alternative hypothesis. For this reason, we subsequently propose a new testing approach that is based on a refined score function that does not require user-selected random terms and entails a test statistic that is normally distributed under both the null and alternative hypotheses.\footnote{An alternative testing approach for conditional independence satisfying Neyman orthogonality that can be applied in our context is suggested by \citet{lundborg2024projectedcovariancemeasureassumptionlean}, 
\citep[see also][]{kook2024algorithmagnosticsignificancetestingsupervised}.}

For the moment, let us assume that the instrument $Z$ is binary. An extension to multivalued IVs is provided in Appendix \ref{app:extensionmulti}. Denote by $p(D,X)=\mathbb{P}(Z=1|D,X)$ the IV propensity score. The score function considered in this case is given below: 
\begin{align}\label{score3}  
&\quad \tilde{\psi}(W,\theta,\eta)\\
&=(\mu(D,X,1)-\mu(D,X,0))^2\notag\\&+2(\mu(D,X,1)-\mu(D,X,0))\left(\frac{(Y-\mu(D,X,1))\cdot Z}{p(D,X)}\notag
-\frac{(Y-\mu(D,X,0))\cdot (1-Z)}{1-p(D,X)}\right)\notag\\
&+\mu(D,X,1)-\mu(D,X,0)+\left(\frac{(Y-\mu(D,X,1))\cdot Z}{p(D,X)}\notag
-\frac{(Y-\mu(D,X,0))\cdot (1-Z)}{1-p(D,X)}\right)\notag\\
&-\theta,\notag
\end{align}
with $\theta_0=\mathbb{E}[\left(\mu(D,X,1)-\mu(D,X,0)\right)^2]+\mathbb{E}[\mu(D,X,1)-\mu(D,X,0)]$.
% For notational convenience, let us define $\theta_0=E[\left(\mu(D,X,Z)-m(D,X)\right)^2]+E[\mu(D,X,Z)-m(D,X)]$.  $\theta_0$ 
Testing based on \eqref{score3} corresponds to an aggregate $L_2$-type measure that can be used to test violations across values of $D$, $X$, and $Z$, which is common in specification tests based on nonparametric regression.\footnote{See e.g.\ \citet{racine1997consistent}, \citet{racine2006testing}, \citet{hong1995consistent} and \citet{wooldridge1992test}.}  In addition to the squared difference in conditional mean outcomes, %that we already encountered, 
the score function notably contains a term in which the difference in conditional mean outcomes is multiplied with a difference in expressions obtained by inverse probability weighting (IPW) with the IV propensity score. In fact, our new score above combines the orthogonalized squared score in \eqref{score2} with the popular doubly robust score. %as also considered for testing in HK. 
Just as the squared difference in conditional mean outcomes, the score function $\tilde{\psi}$ is zero in expectations,
$\mathbb{E}[\tilde{\psi}(W,\theta_0,\eta_0)]=0$,
under the null hypothesis that $\theta_0=0$, which follows from iterated expectations, and 
% \begin{align*}
% &\quad E[\psi(W,\theta_0,\eta_0)]\\
% &=E\left[2(\mu_0(D,X,1)-\mu_0(D,X,0))\left(\frac{(Y-\mu_0(D,X,1))\cdot Z}{p_0(D,X)}
% -\frac{(Y-\mu_0(D,X,0))\cdot (1-Z)}{1-p_0(D,X)}\right)\right]\\
% &=E\left[2E\left[(\mu_0(D,X,1)-\mu_0(D,X,0))\left(\frac{(Y-\mu_0(D,X,1))\cdot Z}{p_0(D,X)}
% -\frac{(Y-\mu_0(D,X,0))\cdot (1-Z)}{1-p_0(D,X)}\right)\right]\Big|Z,D,X\right]\\
% &=0,
% \end{align*}
% where the last equality follows from  $E[Y-\mu_0(Z,D,X)|Z,D,X]=\mu_0(Z,D,X)-\mu_0(Z,D,X)=0$. In addition, 
satisfies the Neyman orthogonality property (see Appendix \ref{proofNeyman}).

When testing, we assume an i.i.d.\ sample of size $n$, in which $i$ is the index of an observation and $W_i=(Y_i,D_i,X_i, Z_i)$ are the variable values of observation $i$ in that sample, with $i$ $\in$ $\{1,2,..,n\}$. We apply cross-fitting as for instance discussed in \citet{doubleML} to avoid over-fitting due to a correlation of the estimation of the nuisance parameters and $\theta_0$.  Therefore, we split the data into $K$ subsamples of size $N=n/K$. The cross-fitted estimator is given by
\begin{eqnarray}\label{eq:test1}
\hat{\theta}=\frac{1}{K}\sum_{k=1}^K\mathbb{E}_{N,k}[\psi_k(W_i,0,\hat{\eta})],
\end{eqnarray} 
with $\hat{\eta}=(\hat{\mu},\hat{p}_1,\dots,\hat{p}_L)$.
Under the regularity conditions in A \ref{assnorm}, $\hat{\theta}$ is asymptotically normal and  $\sqrt{n}$-consistent, as stated in Theorem \ref{theorem1}. The proof is provided in Appendix \ref{proofth1}.
% and is closely related to but not fully equivalent to that in HK, who consider testing based on a binary instrument, while in our test, the instrument may also be multiply discrete or continuous. 
\begin{assumption}[Asymptotic Normality]\label{assnorm}
Define $U=Y-\mu(D,X,Z)$.
The following assumption needs to hold for all $n\ge 3$, $\mathbb{P}\in\mathcal{P}$ and $q>2$: (i) $\|Y\|_{\mathbb{P},q}<C$ and $\mathbb{E}[U^21(Z\in Z_l)]>c$ (ii) Given a random subset $I$ of $[n]$ of size $N = n/K$, the nuisance parameter estimator $\hat{\eta}_0=\hat{\eta}_0((W_i)_{i\in I^c})$ obeys $\|\hat{\eta}-\eta_{0}\|_{\mathbb{P},2q}\le C$,  $\|\hat{\eta}-\eta_{0}\|_{\mathbb{P},4}\le \delta_N$, and $\|\hat{\eta}-\eta_{0}\|_{\mathbb{P},2}\le\delta_N^{1/2}N^{-1/4}$ with $\mathbb{P}$-probability not less than $1-o(1)$.
\end{assumption}
\begin{theorem}\label{theorem1}
Conditional on Assumptions \ref{assnorm}, the estimator in eq. \eqref{eq:test1} satisfies
\begin{eqnarray}\label{eq:tasy}
\sqrt{n}\sigma^{-1}\hat{\theta}\overset{d}\rightarrow N(\theta_0,1),
\end{eqnarray}
uniformly over $P\in\mathcal{P}$, where $\sigma^2 = E[\psi(W,\theta_0,\eta_0)^2]$. Moreover, the result continues to hold if
$\sigma^2$ is replaced by $\hat{\sigma}^2:=\mathbb{E}_n[(\psi(W_i,\hat{\theta},\hat{\eta}))^2]$. Consequently, a test that rejects the null hypothesis $H_0$, $\theta_0=0$, if $|\sqrt{n}\hat{\sigma}^{-1}\hat{\theta}|>\Phi^{-1}(1-\alpha/2)$ has asymptotic level $\alpha$.
\end{theorem}
Theorem \ref{theorem1} states that the test statistic is normally distributed both under the null ($H_0: \theta_0=0$) and alternative hypothesis ($H_1: \theta_0\neq0$). The following Corollary \ref{coro:test} shows that the proposed test is consistent, i.e., the power converges to one as $n\rightarrow \infty$.

\begin{corollary}\label{coro:test}
Let $c_\alpha:=\Phi^{-1}(1-\alpha/2)$ be the critical value of the test proposed above. Under the alternative hypothesis ($\theta_0\neq 0$), it holds
% such that $c_\alpha \rightarrow \infty$ and $c_\alpha = o(\sqrt{n})$, then for $n \rightarrow \infty$:
% \begin{enumerate}
%     \item Under the null hypothesis ($\theta=0)$: $\lim\limits_{n\rightarrow\infty} \, P(\sqrt{n}\hat{\theta}_j/\hat{\sigma}_j < c_\alpha) = 1 $
%     \item Under the alternative hypothesis ($\theta\neq0$): 
    $$\lim\limits_{n\rightarrow\infty}\, P(|\sqrt{n}\hat{\theta}/\hat{\sigma}|> c_\alpha) = 1. $$
% \end{enumerate}
\end{corollary}
\noindent
This holds true since $$P(|\sqrt{n}\hat{\theta}/\hat{\sigma}|> c_\alpha)=P\left(\left|\widehat{\theta}_j\right| \geq c_\alpha \frac{\widehat{\sigma}_j}{\sqrt{n}}\right) \geq P\left(\left|\widehat{\theta}_j-\theta_j\right| \leq \left|\theta_j\right|-c_\alpha \frac{\widehat{\sigma}_j}{\sqrt{n}} \right) =1,$$ 
as long as $c_\alpha=o(\sqrt{n})$. 
%We will need this corollary to state the properties of our testing procedure in the following section.
% \textcolor{blue}{Why not choosing $c_\alpha=\Phi^{-1}(1-\alpha/2)$?}\textcolor{red}{it should be $\theta$ and not $\theta_j$}

%\vspace{-0.4cm}
\section{Selection Method}\label{selmethod}
The tests outlined in Section \ref{testapproach} were conditional on having already defined the instrument $Z$ and covariates $X$ under which the SOO assumption with respect to the treatment supposedly holds. However, a main contribution of this study is a data-driven approach for learning partitions of observed pre-treatment variables into IVs and covariates. To do so, we suggest applying the testing approach iteratively when sequentially considering one variable from the set of all pre-treatment variables, henceforth denoted by $Q$, as instrument $Z$ and the remaining variables as covariates $X$. More specifically, our procedure consists of the following steps (details provided in subsections): 

(1) Select candidate variables with a strong first-stage effect on $D$ from the observed variables $Q$, conditional on remaining variables. $\mathcal{S}$ is the set of strong IVs. The statistical criterion to decide on IV strength will be introduced below. $\hat{S}$ then is the set of candidates selected as strong.

(2) Each candidate in $\hat{S}$ is iteratively defined as the instrument, $Z$, and all remaining variables in $Q$ are defined as covariates, $X$, then the test of hypothesis \eqref{H0_quadratic} is run for each of the candidates. 

(3) If hypothesis \eqref{H0_quadratic} is not rejected in (2) for a candidate, then assign that candidate to the set of IVs for which mean independence holds, $\mathcal{V}$. If there are multiple candidates that pass the test, select the test with the maximal \texttt{p}-value as the final IV and the remaining variables in $Q$ as final covariates $X$. If (2) suggests that the null is violated in all iterations, then implication \eqref{mainresult2} is rejected. 

Step (1) is required to select candidates which satisfy $E[D|X,Z] \neq E[D|X]$, the first stage condition. Let us assume that $Q$ is low-dimensional, meaning that sample size $n$ is larger than the number of variables in $Q$, denoted by $p$. In this case, step (1) might be implemented based on a first stage regression of $D$ on $Q$ and selecting all regressors with statistically significant associations after controlling for multiple hypothesis testing issues into the set of candidate instruments $\hat{S}$. In \textit{high}-dimensional settings where $p>n$, regularization can be applied to select strong IV candidates. 

\subsection{Strong IV Selection}\label{sec:IVselection}
To select strong IVs in the high-dimensional setting, we use first-stage hard-thresholding (FSHT) 
%existing methods that take the treatment model into account. Specifically, we consider the following linear regression model,
proposed by \citet{Guo2018Confidence}. In this approach, IVs are considered irrelevant if their t-statistic does not exceed a predefined threshold. %\footnote{Specifically, an IV is treated as weak if \begin{equation}  |t_{\gamma_j}| = \left| \frac{\hat{\gamma}_j}{\sqrt{\widehat{Var}(\hat{\gamma}_j)}} \right| < \sqrt{2.01 \log(\max(p,n))}. \end{equation} \noindent Here, $\widehat{Var}(\hat{\gamma}_j)$ denotes the variance estimator of the first-stage coefficient estimator $\hat{\gamma}_j$. The right-hand side accounts for multiple testing, where the factor $\log(\max(p,n))$  originates from the tail bound of the normal distribution.
%\textcolor{red}{Is it $L$ here? (also in the following) I think we should ignore $L$ (non-binary instruments) here and use $p=dim(Q)$. Correct?}\textcolor{violet}{Changed everything to $L$, to avoid confusion with the p-value.}
%In the low-dimensional case, this term simplifies to $\sqrt{2.01 \log(n)}$. This procedure is shown to consistently select relevant IVs, as demonstrated in Lemma 1 in \citet{Guo2018Confidence}.
%\citet{Windmeijer2021Confidence} treat the factor 2.01 as a tuning parameter and show in simulations that this choice of the tuning parameter performs well in practice. Meanwhile, \citet{Guo2018Confidence} argue that the consistent selection property holds as long as the argument inside the logarithm grows with $n$.}
%The proposed approaches so far rely on a simple linear regression model.
%To make the effect of the covariates $X$ on $D$ more flexible, %a complementary strategy 
We iteratively consider each variable in $Q$  
%\textcolor{red}{Shouldn't this be $Q$ instead of $Q_j$?} 
as instrument $Z$ and any remaining variables as covariates. That is, when considering the $j$th variable in $Q$ as instrument and denoting it by $Q_j$, we have that $Z=Q_j$ and $X=Q_{[j]}$, where $Q_{[j]} = Q \setminus Q_j$ and $Q = X \cup Z$ when estimating the first-stage association $E[D|X,Z]$. Leveraging the DML literature, we consider the following partially linear specification to estimate the effect of a variable $Z$ on treatment $D$:
\begin{align}
D& =\gamma_j^T Z + g(X) + \varepsilon,\\
Z & = h(X) + v.
\end{align}
Here, $g(X)$ and $h(X)$ are general functions of $X$, and the first-stage effect is indexed by $j$ as it may vary with each variable considered. %\footnote{Note that although the partially linear model assumes additive separability and a linear effect of $Z$ on $D$, estimation based on this model may also have power in fully nonparametric settings, where the effect of $Z$ on $D$ may be nonlinear and may interact with $X$. In such cases, it follows from Stein's Lemma \citep{Stein81} that $\gamma_j$ corresponds to a weighted average of the marginal effects of $Z$ on $D$ given $X$, with all weights being nonzero.} %\footnote{In special cases (in particular $Z$ being Gaussian conditional $X$), $\gamma_j$ is the (unweighted) average of marginal effects.}
DML-based estimation of $\gamma_j$ is asymptotically normal under regularity conditions.\footnote{In particular, estimators of the models of $D$ and $Z$ should attain a convergence rate of $o(n^{-1/4})$. Then $\sqrt{n}\hat{\sigma}_{\gamma,j}^{-1}(\hat{\gamma}_j-\gamma_{0,j})\overset{d}{\rightarrow} N(0,1)$, where $\hat{\sigma}_{\gamma,j}^{-1}$ is the standard error of $\hat{\gamma}_j$. \citep{doubleML}}  
% and the t-statistic for the test with $H0: \gamma_j=0$ also is. 
As only one IV is considered in turn, the first-stage F-statistic is equal to the square of the t-statistic, $t_j^2=F_j$, which is asymptotically $\chi_1^2$-distributed with one degree of freedom, $F_j \overset{d}{\rightarrow} \chi_1^2$ \citep[e.g. Proposition 3 of][]{Masten2021Salvaging}. Under the alternative, $\gamma_j \neq 0$, we have $\frac{F_j}{n} \overset{P}{\rightarrow} \kappa_j$, where $\kappa_j > 0$ is a constant. The critical values $C_n$ for the F-statistic need to satisfy $C_n \to \infty$ and $C_n = o(n)$ as $n \to \infty$, for the FSHT procedure to successfully identify strong IVs with probability approaching 1, $\lim\limits_{n\rightarrow\infty} P(\hat{S} = \mathcal{S}) = 1$. %where $\hat{S}$ is the set of variables selected as strong.  
%This follows from the fact that the square of the t-statistic is equivalent to the F-statistic with one degree of freedom. Asymptotically, $df \cdot F \overset{d}{\rightarrow} \chi^2(df)$, and in this case, $df = 1$.
 %We apply this rule to the FSHT step, i.e., 
An IV $j$ is considered as strong if $F_j>C_{\tilde{\alpha}}$ with $\tilde{\alpha}=0.1/\log(n)$.\footnote{\citet{Windmeijer2021Confidence} exploit a result in \citet{Potscher1983Order} and \citet{Andrews1999Consistent}, which states that a sequence of $\texttt{p}$-values, $\texttt{p}_n$, satisfying $\texttt{p}_n \rightarrow 0$ and $\log(\texttt{p}_n) = o(n)$ can be chosen to meet the just-mentioned conditions on $C_n$.
Based on the recommendation in \citet{Belloni2012Sparse}, they adopt $\texttt{p}_n = 0.1 / \log(n)$.} $C_{\tilde{\alpha}}$ denotes the critical value $q_{\chi^2_1}(1-\tilde{\alpha})$ and $q_{\chi^2_1}(\cdot)$ denotes the quantile function of the $\chi^2_1$-distribution.
%However, \citet{Windmeijer2021Confidence} caution that this approach may be problematic if invalid IVs are more likely to surpass the threshold. %To mitigate this issue, one could follow \citet{Fan2024Endogenous}, who use all candidate instruments in both the selection of relevant and valid instruments.\footnote{\citet{Fan2024Endogenous} focus on the moderately high-dimensional setting, where $\underset{n \rightarrow \infty}{\lim} \frac{\log p_n}{\log n} = c, \quad \text{and } 0 \leq c < 1$. They propose selecting relevant IVs using the adaptive LASSO with weights defined by elastic net estimation. Under the conditions outlined in their Assumption 1, their procedure consistently selects the relevant IVs. In terms of model flexibility, the caveat of their method is that their approach relies on linearity.}
As one way to implement DML, which is particularly useful in high dimensions ($p>n$), we estimate the nuisance functions $g(X)$ and $h(X)$ using the LASSO. %Consequently, we apply the CORTH Features procedure from \citet{soleymani2022causal} to the first stage estimation to detect the parent variables of $D$ (those directly affecting the treatment), in other words, the relevant IVs. %We note that the use of DML allows us to deal with high-dimensional sets of covariates. %Given that the Such doubly robust approaches are also applicable in high-dimensional contexts with $n << p$ when combining them with ML to control for important control variables among high-dimensional covariates $X$. 
\subsection{Valid IV Selection and Identification Test}
In step (2), we test the null hypothesis in eq. \eqref{H0_quadratic} iteratively over all candidate IVs that pass the first-stage threshold. 
Our aim is to find a partition of variables for which the conditional independence of the respective candidate IV holds.
 To discuss this more formally, we introduce the partition $$\mathcal{P}_j = \{Z=Q_j, \ \ X=Q_{[j]}\},$$ such that variable $j$ in set $Q$ is chosen to be the IV, while the remaining variables in $Q$ are used as controls. Moreover, let $\mathcal{V}$ denote the set of candidate IVs which are conditionally mean independent of the outcome, satisfying condition \eqref{mainresult2} and the null in eq. \eqref{H0_quadratic}:
$$\mathcal{V} = \{j: E[Y|D,X]=E[Y|D,X,Z]\}.$$
Moreover, we denote the set of partitions for which the IV has a first stage effect on the treatment and the conditional independence of the IV holds by $\mathcal{P}^*$: 
%\vspace{-1cm}\\
\begin{equation}\label{eq:passes}
	\mathcal{P}^* =  \{\mathcal{P}_j: j \in \left(\mathcal{S} \cap \mathcal{V}\right) \}.
\end{equation}
%\vspace{-1cm}\\
The corresponding estimated set of partition(s) $\mathcal{P}^*$, denoted by $\hat{\mathcal{P}}_{pass}$, is given by 
\begin{equation}\label{eq:selected}
	\hat{\mathcal{P}}_{pass} = \left\{\mathcal{P}_j: j \in \left(\hat{S} \cap \hat{\mathcal{V}}\right) \right\},
\end{equation}
where $\hat{\mathcal{V}}=\{j: |\sqrt{n}\hat{\sigma}^{-1}_j\hat{\theta}_j| < c_\alpha\}$ is the set of instruments for which conditional independence is not rejected based on the test defined in Theorem \ref{theorem1} and Corollary \ref{coro:test} with significance level $\alpha$ and critical value $c_{\alpha}$.
% in $\hat{S}$. \sqrt{n}\|\{\hat{\sigma}^{-1}_j\theta_j(\mathcal{P})\}_{m=1}^M\|_p < c_\alpha\} The index $m=1,...,M$ refers to the instruments, with $M$ being the number of instruments. 
% In the following, we focus on the case that there are only binary instruments $Z$ in the set $\hat{S}$. 
%Even more crucial than the question of how to choose the final partition if $|\hat{\mathcal{P}}_{pass}| > 1$ is whether $\hat{\mathcal{P}}_{pass}= \emptyset$ or not. 
Our main contribution is the development of a new procedure which tests the identification of a causal effect in a data-driven way. Thus, we consider the following hypotheses:
%\vspace{-1cm}
\begin{align}\label{identification}
H_0: \text{no identification}\quad vs. \quad H_1: \text{identification (conditional mean independence}).
\end{align}
%\vspace{-1.4cm}

%\vspace{-1cm}
We conclude that $H_1$ is true (the ATE is identified, see the null hypothesis in \eqref{nullhyp1}) if $\hat{\mathcal{P}}_{pass}\neq \emptyset$. 
The following Theorem helps us understand the type 1 error of our test.
%\vspace{-1.2cm}
\begin{theorem}\label{prop:all}
    Under the assumptions of Th. \ref{theorem1}, assuming $\lim\limits_{n\rightarrow\infty} P(\hat{S}=\mathcal{S}) = 1$, for a given $\alpha$, it holds %\vspace{-1cm}
$$\lim\limits_{n\rightarrow\infty}P(\hat{\mathcal{P}}_{pass}\subseteq\mathcal{P}^*)=1.$$
\end{theorem}
%\vspace{-1cm}

%\FloatBarrier

 % The assumptions on the critical values are similar to those made for consistency of the selection procedure in \citet{Andrews1999Consistent}.
 Theorem \ref{prop:all} states that $\lim\limits_{n\rightarrow\infty}P(\mathcal{V}\neq\emptyset)=1$
 if $\hat{\mathcal{V}}\neq \emptyset$. Hence, the type 1 error of our proposed identification test in \eqref{identification} goes to zero as $n\rightarrow \infty$. This means that if our test finds identification ($\hat{\mathcal{P}}_{pass}\neq \emptyset$), there is identification with probability $1$ for large $n$. The next theorem provides insights about the type 2 error of our test, i.e., how likely it is that we can find identification if there is identification. Theorem \ref{power} states that the type 2 error is at least bounded by $\alpha$.
 \begin{theorem}\label{power}
    Assume that $\mathcal{P}^*\neq\emptyset$ (testable identification). Under the assumptions of Theorem \ref{theorem1}, assuming $\lim\limits_{n\rightarrow\infty} P(\hat{S}=\mathcal{S}) = 1$, for a given $\alpha$, it holds
$$\lim\limits_{n\rightarrow\infty}P(\hat{\mathcal{P}}_{pass}\neq\emptyset)\ge 1-\alpha.$$
\end{theorem}

If $|\hat{\mathcal{P}}_{pass}| > 1$ such that there is more than one sufficiently strong candidate IV while conditional independence is not rejected, 
\iffalse
we can proceed in two ways. For a set of partitions $\hat{\mathcal{P}}_{pass} = \{\mathcal{P}_1, ..., \mathcal{P}_{|\hat{\mathcal{P}}_{pass}|},\}$ where each partition is associated with a specific IV, we define the set of IVs associated with that set of partitions as 
\begin{equation}\hat{\mathcal{I}} = \left\{j: \mathcal{P}_j \in \hat{\mathcal{P}}_{pass}\right\},
\end{equation}
we can then select the final partition into IVs and controls as 
\begin{equation}\label{eq:all}
\hat{\mathcal{P}}_{all} = \{Z=Q_{\hat{\mathcal{I}}}, \ \ X=Q \setminus Q_{\hat{\mathcal{I}}}  \}.
\end{equation}
\fi
% \textcolor{red}{$P_{pass}=P_{all}$?? I do not see why we need the definition in 4.6 and 4.7} \textcolor{blue}{Here, we propose two ways to select the final partition: (1) \textit{all} IVs that pass, treating the other variables as controls or (2) the IV associated with the largest p-value. It is true that this is not what we do in the end, but it is one possibility and it seems more in line with the title / the original goal, i.e. learning IVs and controls (arguably we fall a bit short of that now). $P_{pass}$ is a set of partitions and $P_{all}$ is itself a partition. But if it is clear what we mean, we can also drop these definitions.}
we select the final partition as the one which maximizes the \texttt{p}-value when testing conditional independence:%\footnote{In principle, one could also select all IVs which belong to $\hat{\mathcal{P}}_{pass}$ as valid.}
\begin{equation}\label{eq:pmax}
\hat{\mathcal{P}}_{pmax} = \underset{\mathcal{P}_j \in \hat{\mathcal{P}}_{pass}}{argmax} \ \ \texttt{p}(\mathcal{P}_j), 
\end{equation}
where $\texttt{p}(\mathcal{P}_j)$ is the $\texttt{p}$-value obtained in the conditional independence test of $Z$ and $Y$. 
Finally, Algorithm \ref{algo:procedure} in the appendix describes the steps of our method by means of pseudo-code and we have added a discussion of the computational cost in appendix \ref{app:computationalcost}. 

\section{Simulation study}\label{simulation}

In this section, we briefly summarize the results of the simulation study. The detailed settings and the full results can be found in Appendix \ref{app:simulation}. 
We consider two main settings: one with a single valid IV and one with multiple valid IVs. These are then again split into two settings, one with binary and one with continuous instruments. The outcome equation is linear, while we choose a linear index model to generate $D$. We encode the violations of Assumptions \ref{ass3} and \ref{ass2} via parameters in this two-stage model. Variables in $Q$ are correlated and error terms are standard normal. We then vary the sample sizes, setting them to $n=1000, 4000$ and $16,000$. We use the LASSO for all ML steps, to illustrate our algorithm.

Our results can be summarized as follows: first, we observe that with no violations of SOO and validity, as $n$ increases, the probability of $\hat{\mathcal{P}}_{pass}$ being non-empty, finding identification, increases. The probability of correctly selecting a valid IV in the final partition increases, while that of incorrectly selecting a confounder decreases. These results are especially clear when multiple valid IVs are available. When there is no violation, we expect the violation parameter, $\hat{\theta}$, to be close to zero over the repetitions, and this is also what we observe in the results. 
Secondly, we model violations of SOO (A4), through an unobserved confounder, and IV validity (A5), through a violation of the exclusion restriction. When one of these is violated, across all settings the probability of $\hat{\mathcal{P}}_{pass}$ being empty quickly increases with sample size. The latter results are particularly clear when the violation parameter is reasonably large as is the case with violations of IV validity in our simulation. 

Overall, our procedure consistently selects the correct IV(s) when Assumptions \ref{ass1} and \ref{ass2} are satisfied, and the algorithm correctly concludes there is no identification in the case of violations of SOO or IV validity. The method seems to perform particularly well when $n$ is large, when violations are clearly separated from zero and when there are multiple candidate IVs. 

%Similar results are seen in simulations for the continuous IV version.\footnote{Available upon request from the authors; simulations ongoing.}

%\spacingset{1}
\begin{table}[htbp]
\begin{center}
\caption{Empirical Application}
	\label{tab:ohi}
    \begin{footnotesize}
\begin{tabular}{lccc}
  \hline
  \multicolumn{4}{l}{\textbf{PANEL A: Primary Care Visits}} \\
Method & $\hat\theta$ & se & \texttt{p}-value \\ 
  \hline
  LASSO & 0.000 & 0.000 & 0.997 \\ 
  Random Forest & -0.000 & 0.013 & 0.999 \\ 
  XGBoost & -0.070 & 0.207 & 0.735 \\ 
  \hline
    \multicolumn{4}{l}{\textbf{PANEL B: Number of Prescriptions}} \\
  \hline
LASSO & 0.001 & 0.000 & 0.148 \\ 
  Random Forest & -0.007 & 0.010 & 0.472 \\ 
  XGBoost & 0.073 & 0.088 & 0.407 \\ 
   \hline
\end{tabular}
\end{footnotesize}
\end{center}
\par
{\scriptsize Notes: this table reports the estimate $\hat\theta$ from eq. \eqref{eq:test1} with five folds ($K=5$) when using the doubly robust score functions in eqs. \eqref{score3} and \eqref{score4}. `se` and `\texttt{p}-value` report the standard error and reported \texttt{p}-value for estimate $\hat\theta$ on the random program assignment variable.  
}
\end{table}
%\spacingset{1.8}

\vspace{-0.5cm}
\section{Empirical application}\label{application}

We apply our method to the Oregon Health Insurance Experiment, in which low-income adults were randomly selected by lottery to be eligible to apply for Medicaid. Oregon launched the lottery in early 2008 for about 90,000 participants, with notifications through October 2009. Previous work uses random assignment as an IV for insurance coverage and finds increases in healthcare utilization, reductions in out-of-pocket expenditures, and improvements in self-reported health.\footnote{See \citet{finkelstein2012oregon, baicker2013oregon, taubman2014medicaid, finkelstein2019value}}

We apply our testing methodology to the experimental OHI data, using the sample definition adopted by \citet{finkelstein2012oregon}, which yields 23,762 observations. The treatment $D$ indicates whether an individual ever enrolled in Medicaid by October 2009. While assignment is randomized, enrollment may be selective due to non-compliance. %For instance, some individuals selected as eligible for Medicaid may choose not to enroll. 
 We consider two outcomes: the number of primary care visits and the number of prescription medications. Along with random assignment as a natural IV candidate, the vector 
$Q$ contains 218 pre-assignment characteristics, including demographics, socioeconomic variables, and pre-treatment health, utilization, and expenditure measures, as well as indicators for missing values.

Testing is based on the cross-fitted estimator $\hat\theta$ of eq. \eqref{eq:test1} with five folds ($K=5$) when using the doubly robust score functions in eqs. \eqref{score3} and \eqref{score4} in the case of binary and continuous IVs, respectively. In the case of a continuous candidate IV, we make use of the score function outlined in Appendix \ref{app:extensionmulti} for multivalued IVs. We partition the support of the continuous candidate IV using quartiles. As in the simulations, nuisance functions are estimated using LASSO, with random forest and gradient boosting as alternatives. We drop observations with extremely low or high propensity scores, which cause estimation instability and lead to high variance.\footnote{For a discussion of such trimming based on the propensity score, see, e.g., \citet{crump2009dealing} and \citet{lechner2019practical}. Specifically, we discard observations for which the estimated propensity score \( p_l(D, X) = P(Z \in Z_l \mid D, X) \), defined for a partition \( Z_l \) of the candidate IV (whether continuous or discrete), falls outside the interval \( 0.01 < p_l(D, X) < 0.99 \). We require that no more than 5\% of observations be trimmed using this rule in order for the candidate IV to be considered.} 

Results are reported in Table \ref{tab:ohi}. The algorithm selects random program assignment as the only valid IV, with $\hat\theta$ close to zero. 
Although multiple candidates enter $\hat{S}$ after the first stage, none satisfy the trimming requirement except random assignment.\footnote{We report every variable selected in the first stage to be included in $\hat{S}$, sorted by decreasing \texttt{p}-value, in Appendix Table \ref{tab:application_allcovar}, as well as the number of observations dropped following the trimming rule.} Using the 30\% threshold, random assignment passes the test across all learners for the doctor visits outcome. For LASSO, the estimate is zero with a \texttt{p}-value of 99.7\%, supporting both IV validity and the SOO assumption given the observed covariates. On the one hand, this indicates validity of the IV, implying not only the randomness of program assignment inherent to the experimental design, but also that the assignment does not directly affect the earnings outcome other than through the treatment (for example, through motivation or disappointment when being or not being eligible for the program). On the other hand, the testing result suggests that the SOO assumption holds for the treatment when controlling for the pre-assignment covariates available in the data. For the prescriptions outcome, results are less conclusive: random forest and XGBoost yield \texttt{p}-values above 30\%, while LASSO does not. Consequently, given the covariates, our result suggests that we may evaluate the average treatment effect (ATE) on the total population for the doctor visits outcome. In contrast, employing an IV-based approach utilizing random assignment as the IV for effect estimation would, under specific additional assumptions like treatment monotonicity in the IV, only permit assessing the local average treatment effect (LATE) on the subpopulation of compliers, whose training participation aligns with the random assignment \citep{Imbens+94}. 

%For the prescriptions outcome, results are less conclusive: random forest and XGBoost yield p-values above 0.3, while LASSO does not. Overall, the evidence supports identification of the ATE for doctor visits. By contrast, an IV estimator based on random assignment would identify a LATE for compliers under additional assumptions such as monotonicity \citep{Imbens+94}.

\section{Conclusion}\label{co}

In this paper, we introduced an ML-based algorithm based on a novel doubly robust score function designed to detect and test, in a data-adaptive manner, the presence of control variables sufficient for identifying treatment effects in observational data, as well as variables satisfying IV validity. Treatment effects may be heterogeneous, but identification relies on a conditional mean restriction and therefore pertains to the ATE rather than the full distribution of treatment effects. %If the algorithm concludes that identification is given, the researcher may proceed to estimate the ATE, but estimation of the LATE or conditional ATEs are also an option. 
The method searches over partitions of variables into candidate controls and an IV and tests a conditional mean-independence implication that must hold when the SOO assumption and IV validity jointly hold true. This represents an important advancement over previous work which relies on a priori assumptions about whether a variable is an IV or a control. %ly suggested tests of identifying assumptions for causal inference in observational data, which rely on a priori assumptions about whether a variable is a control variable or an IV. 
We demonstrated that the method consistently detects controls and IVs (if they exist) both through theory and simulations. Moreover, an application to the OHI Experiment confirms that our algorithm correctly selects random assignment into the program as IV, across various ML algorithms for the nuisance function learners and in line with what a researcher would expect. This shows that our new algorithm can be applied as explorative method, where the researcher wants to learn the partition from the data, or as a confirmatory method. % the consistency of the method for detecting control variables and IVs (if they exist) %under regularity conditions 
%and investigated its finite sample performance through a simulation study, which supports the consistency result. %However, as a word of caution for empirical applications, the results also suggest that decent performance of the test might require a relatively large dataset rather with more than just several 1000 observations. Further experiments (not reported in this work) improved the performance of the test in terms of reducing the rate at which a confounder is chosen as the candidate IV, when the strength of violations of Assumptions \ref{ass1} and \ref{ass2} are increased by increasing the values of $\gamma$ and $\delta$. However, as the strength of these violations may not be so severe in practice, we maintain lower values for $\gamma$ and $\delta$ in our reported simulation results. 
%Finally, we applied our algorithm to  empirical health data from the OHI Experiment. Our approach (as expected) always selected random assignment into the program as valid IV, across a variety of algorithms for the nuisance function learner.

\bibliographystyle{apalike}
\bibliography{research.bib}

\appendix

\section{Literature review}\label{app:litrev}

Our paper contributes to a growing literature on testing identifying assumptions for causal inference. For instance, \citet{deLunaJohansson2012} and \citet{BlackJooLaLondeSmithTaylor2015} make use of the same conditional independence of the IV as considered here to test the SOO assumption for the treatment when assuming a valid IV, based on matching or regression estimators.\footnote{See also the related test by \citet{chen2018testing}, which (in contrast to other methods) requires symmetrically distributed error terms. Furthermore, \citet{Angrist2004}, \citet{BrinchMogstadWiswall2012}, and \citet{Huber2013} assume IV validity to hold unconditionally without controlling for covariates, in order to test the unconditional independence of the treatment and potential outcomes. \citet{BertanhaImbens2015} consider the fuzzy regression discontinuity design, where IV validity holds at a cutoff of a running variable which discontinuously affects treatment assignment.} 
%That is, they test one assumption (SOO) conditional on the other one (IV validity) based on matching or regression estimators.
 In contrast, HK jointly test both IV validity and SOO assumptions with pre-defined instruments and covariates using ML approaches which permit for high-dimensional covariates.\footnote{\citet{angrist2017leveraging} also propose a joint test for parametric models with low-dimensional covariates.}  %\footnote{\citet{CaetanoCaetanoFeNielsen2021} provide a related test for a treatment with multiple values and bunching at the highest or lowest value. For instance, schooling laws may impose a minimum number of years of education, such that all subjects who would otherwise have acquired a lower level of education bunch at this minimum. For this reason, also the subjects' unobserved characteristics affecting the treatment bunch at this point. One can test the SOO assumption by checking whether a dummy variable for the bunching point, which is a function of the unobserved characteristics, is conditionally independent of the outcome given the treatment and the covariates. Therefore, the dummy variable at the bunching point of the unobservables has the same role as the instrument in HK, given that the outcome model as a function of the treatment and covariates is correctly specified.} 
 This is conceptually closest to our approach, but one important difference is that HK require specifying the sets of supposed IVs and covariates to be used for testing, whereas in this paper, these sets of covariates and IVs are learned from the data and may therefore be a priori unknown.

Moreover, we provide a new orthogonalized quadratic score, which improves upon the quadratic score used in HK. This new score is of broader interest since it can be applied in many other contexts beyond ours. One example is the study by \citet{parikh2024double}, who compare experimental and non-experimental treatment effect estimates and use an indicator to distinguish between experimental and non-experimental evaluation designs, which plays a role comparable to the IV in our paper. Imposing external validity of the experiment allows to assess the SOO assumption in the non-experimental design. Conversely, imposing the SOO assumption in the non-experimental design allows to assess the external validity of the experiment, which is equivalent to testing IV validity in our context. Unlike our paper, the authors do not consider testing both assumptions jointly.\footnote{Similarly, \citet{angrist2015wanna} employ the same conditional independence of the IV as considered here to test IV validity within the framework of the sharp regression discontinuity design, where SOO for the treatment holds by design, as it is a deterministic function of a cutoff in a running variable. This allows testing whether the running variable is a valid IV, i.e., whether it is not associated with the outcome conditional on the treatment.}%\footnote{If the running variable is a valid instrument, effects can also be identified away from the cut-off, which is otherwise not feasible due to the lack of common support in the treatment across different values of the running variable.} 

Relatedly, one strand of the statistics literature imposes the SOO assumption in observational studies and exploits the conditional independence condition to identify subsets of covariates that are sufficient for identification, implying that the remaining covariates satisfy IV validity; see, e.g, \citet{de2011covariate} and \citet{vanderweele2011new}. Considering subsets of covariate information may also allow for more efficient causal effect estimation. For example, \citet{christgau2024efficientadjustmentcomplexcovariates} propose a deep learning-based estimator that learns efficient covariate representations from unstructured data and is asymptotically normal. Again, our approach differs from these studies in that it tests the SOO and IV validity assumptions jointly, rather than testing one and assuming the other. Closer to us, \citet{entner2013data} do not pre-impose the SOO assumption when searching for a sufficient set of covariates to control for based on conditional independence and consider a parametric modelling approach for testing, however, without asymptotic guarantees. Here, we suggest a ML-based procedure that allows for nonparametric models and potentially high-dimensional covariates, while having desirable asymptotic properties. %under certain regularity conditions.

Our study is also related to several contributions in the literature on causal discovery.\footnote{See, e.g.,\ \citet{KalischB2014}, \citet{peters2017elements}, \citet{Glymouretal2019} for reviews.} \citet{peters2015causal} make use of pre-defined IVs to learn which of the observed variables are treatments in the sense that they directly affect the outcome, assuming that these treatments satisfy the SOO assumption. The approach makes use of IVs in a way that may entail the rejection of treatments which violate the SOO assumptions, thereby providing power to detect identification failures. Our approach is different in that it focuses on a single, pre-defined treatment of interest, tests SOO and IV validity jointly, and learns the sets of covariates and IVs from the data. 

\citet{soleymani2022causal} and \citet{quinzan2023drcfs} provide algorithms that select treatments, while also controlling for observed covariates based on the double machine learning (DML) framework of \citet{doubleML}.\footnote{The idea is to sequentially consider each of the observed variables as treatment variable, while considering all remaining variables as covariates to estimate the direct effect of each candidate treatment on the outcome by DML. The algorithm retains only those variables that exhibit statistically significant effects on the outcome.} %(where judging statistical significance should account for issues related to testing multiple hypotheses based on multiple candidate treatments). 
In contrast to our approach, these algorithms do not exploit IVs to test identifying assumptions, but assume the SOO assumption holds for all treatments.

Another domain of causal discovery related to our study is Y-learning \citep[e.g. ][]{SevillaMayn2021}. Conditional on covariates, Y-learning implies that if two variables are independent of each other when not controlling for the treatment, statistically associated with each other when controlling for the treatment, and  both independent of the outcome when controlling for the treatment, then these two variables are relevant and valid IVs. A further implication is that the SOO assumption holds. Our approach differs from Y-learning in that it imposes more causal structure by assuming that potential IVs and covariates are not affected by the treatment. For this reason, our method only requires a single (and a priori unknown) instrument, while Y-learning hinges on the existence of (at least) two IVs. %Nevertheless, the algorithms are useful also in our context, namely for selecting potential instruments with a sufficiently strong first stage effect on the pre-defined treatment conditional on other observed variables.

Another strand of the IV validity and causal discovery literature exploits testable implications of structural causal models. These implications arise either from rank (tetrad/trek) constraints on the covariance matrix in linear settings, or from distributional assumptions such as non-Gaussianity of the errors. \citet{Kuroki2005Instrumental} use tetrad (rank) constraints to derive testable implications of IV models. Their approach requires multiple valid IVs (in particular, overidentifying tetrad restrictions arise when at least three valid IVs are available) and the resulting constraints are necessary but not generally sufficient for validity. \citet{Silva2017Learning} propose IV discovery methods based on tetrad constraints and non-Gaussianity in linear non-Gaussian acyclic models (LiNGAM). Their IV-TETRAD procedures require at least two valid IVs to operate and therefore do not apply in settings where only a single IV is valid. Their algorithms return equivalence classes of causal effects. \citet{Xie2022Testability}, in contrast, derive testable constraints in linear non-Gaussian acyclic models that can rule out invalid IVs. 
In contrast to these papers, \citet{Guo2024Testability} consider an additive non-linear non-constant effects model with non-Gaussian errors when exogenous covariates are present. They propose the auxiliary-based independence test condition, which provides a necessary condition for IV validity and becomes sufficient under additional assumptions. \citet{Wiedermann2026Testing} shows how to test confoundedness of a single IV in a linear model with non-Gaussian errors, but can not test the exclusion restriction. Our method is different from these approaches in that we do not make assumptions on the outcome model or the distribution of the errors and use a test based on mean independence. Moreover, we provide a general, doubly robust, machine learning-based score to implement the test. %This is different from our approach in that we are able to test jointly the validity of a single IV and unconfoundedness.

Our approach is not intended as a general causal discovery or structure learning method. The causal ordering and structural assumptions are taken as given based on the researcher’s substantive knowledge. Within this fixed structure, the procedure tests whether the data support identification of the ATE via the SOO assumption and the existence of a valid instrument.

The computer science literature closest to us is on representation learning of variable sets. \citet{hassanpour2019learning} and \citet{wu2021learning}, for instance, propose deep learning algorithms minimizing global loss functions to simultaneously decompose pre-treatment variables into IVs, confounders, and outcome predictors. Yet, there are several differences between their studies and ours: First, their approaches impose SOO a priori to isolate IVs (whose exclusion from treatment effect estimation can increase efficiency) based on the same conditional independence condition as considered in our paper. However, we exploit the conditional independence condition to test SOO and IV validity jointly. Second, their algorithms, which minimize a global loss function, aim at learning the full set of IVs, which is attractive for maximising efficiency in treatment effect estimation but might be very ambitious to achieve in a finite sample. In contrast, our algorithm pursues the more modest goal of detecting at least one valid IV—a sufficient condition for identification. Third, we demonstrate the consistency of our algorithm in selecting valid IVs under certain regularity conditions, while the asymptotic behaviour of the deep learning methods in \citet{hassanpour2019learning} and \citet{wu2021learning} has not been derived.

We also contribute to a growing literature in statistical learning and econometrics that tries to separate valid from invalid IVs, see for instance \citet{Kang2016Instrumental}, \citet{Guo2018Confidence}, \citet{Windmeijer2021Confidence}, \citet{Windmeijer2019}, \citet{apfel2024agglomerative}, and \citet{apfel2022detecting}. These approaches rely on the assumption that a majority or plurality of (a priori unknown) IVs  is valid. In order to detect them Sargan-type tests in combination with IV-based estimators are used. In contrast, our approach does not pre-impose the existence of valid IVs, but tests validity (for a single IV) and SOO assumptions jointly, in order to apply estimation based on the SOO assumption. Moreover, these methods impose a linear model, which we do not require in our method.

Finally, other recent papers also rely on the majority and plurality assumptions from the literature discussed in the preceding paragraph. 
Kuang et al. (2020) propose Ivy, a method to synthesize information from multiple possibly weak and invalid IVs into a single summary IV, relying on the majority assumption. Hartford et al. (2021) propose modeIV, which uses a plurality assumption in nonlinear models to aggregate IV estimates that cluster around a modal value and therefore relies on homogeneity to interpret deviations as violations of validity.
While our procedure also involves sequential testing steps, its objective is fundamentally different from iterative IV screening or aggregation methods such as Ivy and modeIV. Our approach is not designed to construct an IV estimator, but to test whether the identifying assumptions required for ATE estimation via covariate adjustment are supported by the data. Accordingly, we do not rely on majority or plurality assumptions, nor do we assume homogeneity. Instead, we test whether there exists at least one valid instrument that jointly supports instrument validity and the selection-on-observables assumption.
%\citet{Kuang2020Ivy} propose Ivy, a method to synthesize information from multiple possibly weak and invalid IVs into a single summary IV, relying on the majority assumption. \citet{Hartford2021Valid} propose modeIV, which uses a plurality assumption in nonlinear models to aggregate IV estimates that cluster around a modal value. A homogeneity assumption is needed as deviations from the modal estimate are interpreted as a violation.  This approach additionally relies on a homogeneity assumption, as deviations from the modal estimate are interpreted as violations of validity, and does not accommodate the coexistence of heterogeneous treatment effects (LATEs) and invalid instruments. 
%While our procedure also involves sequential testing steps, its objective is fundamentally different from iterative IV screening or aggregation methods. Our approach is not designed to construct an IV estimator, but to test whether the identifying assumptions required for ATE estimation via covariate adjustment are supported by the data. In particular, we do not rely on majority or plurality assumptions, nor do we assume homogeneity. Instead, we test whether there exists at least one valid instrument that jointly supports instrument validity and the selection-on-observables assumption. In this sense, instruments play a diagnostic role rather than serving as the basis for estimation.

\section{Proof of moment condition and Neyman orthogonality of $\tilde{\psi}$}\label{proofNeyman}

Equation \eqref{score3} suggests the following score function for testing when $Z$ is binary: 
\begin{align}\label{score3a}
&\quad \tilde{\psi}(W,\theta,\eta)\\
&=(\mu(D,X,1)-\mu(D,X,0))^2\notag\\&
+2(\mu(D,X,1)-\mu(D,X,0))\left(\frac{(Y-\mu(D,X,1))\cdot Z}{p(D,X)}\notag
-\frac{(Y-\mu(D,X,0))\cdot (1-Z)}{1-p(D,X)}\right)\\
&+\mu(D,X,1)-\mu(D,X,0)+\left(\frac{(Y-\mu(D,X,1))\cdot Z}{p(D,X)}\notag
-\frac{(Y-\mu(D,X,0))\cdot (1-Z)}{1-p(D,X)}\right)\notag\\
&-\theta\\
&:=\tilde{\psi}_1(W,\theta,\eta)+\tilde{\psi}_2(W,\theta,\eta)-\theta\notag
\end{align}
with
\begin{align*}
&\quad\tilde{\psi}_1(W,\theta,\eta)\\
&=(\mu(D,X,1)-\mu(D,X,0))^2\notag\\&+2(\mu(D,X,1)-\mu(D,X,0))\left(\frac{(Y-\mu(D,X,1))\cdot Z}{p(D,X)}\notag
-\frac{(Y-\mu(D,X,0))\cdot (1-Z)}{1-p(D,X)}\right)
\end{align*}
and
$$\tilde{\psi}_2(W,\theta,\eta)=\mu(D,X,1)-\mu(D,X,0)+\left(\frac{(Y-\mu(D,X,1))\cdot Z}{p(D,X)}\notag
-\frac{(Y-\mu(D,X,0))\cdot (1-Z)}{1-p(D,X)}\right)\notag.$$
The moment condition $\mathbb{E}[\tilde{\psi}(W,\theta_0,\eta_0)]=0$ holds, because 
\begin{align*}
\mathbb{E}\left[(\mu_0(D,X,1)-\mu_0(D,X,0))^2+
(\mu_0(D,X,1)-\mu_0(D,X,0))\right]
-\theta=0
\end{align*}
and
\begin{align*}
\mathbb{E}\left[\left(\frac{(Y-\mu_0(D,X,1))\cdot Z}{p_0(D,X)}\notag
-\frac{(Y-\mu_0(D,X,0))\cdot (1-Z)}{1-p_0(D,X)}\right)\right]=0, 
\end{align*}
as
\begin{align*}
\mathbb{E}\left[\frac{(Y-\mu_0(D,X,1)) Z}{p_0(D,X)}\right]&=\mathbb{E}\left[\frac{Z}{p_0(D,X)}\mathbb{E}\left[Y-\mu_0(D,X,1)|D,X,Z\right]\right]\\
&=P(Z=1)\mathbb{E}\left[\frac{1}{p_0(D,X)}\left(\mathbb{E}\left[Y|D,X,Z\right]-\mu_0(D,X,1)\right)\bigg|Z=1\right]\\
&=P(Z=1)\mathbb{E}\left[\frac{1}{p_0(D,X)}\left(\mathbb{E}\left[Y|D,X,Z=1\right]-\mu_0(D,X,1)\right)\bigg|Z=1\right]=0
\end{align*}
and analogously $\mathbb{E}\left[\frac{(Y-\mu_0(D,X,0))\cdot (1-Z)}{1-p_0(D,X)}\right]=0$.
Furthermore, by the same argument, we have
\begin{align*}
\mathbb{E}\left[(\mu_0(D,X,1)-\mu_0(D,X,0))\left(\frac{(Y-\mu_0(D,X,1))\cdot Z}{p_0(D,X)}
-\frac{(Y-\mu_0(D,X,0))\cdot (1-Z)}{1-p_0(D,X)}\right)\right]=0.
\end{align*}
Neyman orthogonality of $\tilde{\psi}$ can be shown by taking the Gateaux derivates w.r.t.\ the nuisance parameters: 
\allowdisplaybreaks{
\begin{align*}
&\quad\partial_rE[\tilde{\psi}_1(W,\theta_0,\eta_0 + r(\eta-\eta_0)]\big|_{r=0}\\
&=E\left[\partial_r\tilde{\psi}_1,(W,\theta_0,\eta_0 + r(\eta-\eta_0))\big|_{r=0}\right]\\
&=2E\left[((\mu_0(D,X,1)-\mu_0(D,X,0))((\mu(D,X,1)-\mu_0(D,X,1))-(\mu(D,X,0)-\mu_0(D,X,0))) \}\right]\\
&\quad-2E\left[(\mu_0(D,X,1)-\mu_0(D,X,0))\frac{Z(\mu(D,X,1)-\mu_0(D,X,1))}{p_0(D,X)}\right]\\
&\quad+2E\left[(\mu_0(D,X,1)-\mu_0(D,X,0))\frac{(1-Z)(\mu(D,X,0)-\mu_0(D,X,0))}{1-p_0(D,X)}\right]\\
&\quad-2E\left[(\mu_0(D,X,1)-\mu_0(D,X,0))\frac{Z(Y-\mu_0(D,X,1))(p(D,X)-p_0(D,X))}{p_0(D,X)^2}\right]\\
&\quad-2E\left[(\mu_0(D,X,1)-\mu_0(D,X,0))\frac{(1-Z)(Y-\mu_0(D,X,0))(p(D,X)-p_0(D,X))}{(1-p_0(D,X))^2}\right]\\
&\quad+2E\left[(\mu_0(D,X,1)-\mu_0(D,X,0))\frac{(1-Z)(Y-\mu_0(D,X,0))(p(D,X)-p_0(D,X))}{(1-p_0(D,X))^2}\right]\\
&\quad+2E\bigg[((\mu(D,X,1)-\mu_0(D,X,1))-(\mu(D,X,0)-\mu_0(D,X,0)))\\
&\quad\left(\frac{(Y-\mu_0(D,X,1))\cdot Z}{p_0(D,X)}
-\frac{(Y-\mu_0(D,X,0))\cdot (1-Z)}{1-p_0(D,X)}\right)\bigg]\\
&=2E\left[((\mu_0(D,X,1)-\mu_0(D,X,0))((\mu(D,X,1)-\mu_0(D,X,1))-(\mu(D,X,0)-\mu_0(D,X,0))) \}\right]\\
&\quad-2E\left[(\mu_0(D,X,1)-\mu_0(D,X,0))\frac{Z(\mu(D,X,1)-\mu_0(D,X,1))}{p_0(D,X)}\right]\\
&\quad+2E\left[(\mu_0(D,X,1)-\mu_0(D,X,0))\frac{(1-Z)(\mu(D,X,0)-\mu_0(D,X,0))}{1-p_0(D,X)}\right]\\
&=2E\left[((\mu_0(D,X,1)-\mu_0(D,X,0))((\mu(D,X,1)-\mu_0(D,X,1))-(\mu(D,X,0)-\mu_0(D,X,0))) \}\right]\\
&\quad-2E\left[(\mu_0(D,X,1)-\mu_0(D,X,0))(\mu(D,X,1)-\mu_0(D,X,1))\mathbb{E}\left[\frac{Z}{p_0(D,X)}\Bigg|D,X\right]\right]\\
&\quad+2E\left[(\mu_0(D,X,1)-\mu_0(D,X,0))(\mu(D,X,0)-\mu_0(D,X,0))\mathbb{E}\left[\frac{(1-Z)}{1-p_0(D,X)}\Bigg|D,X\right]\right]\\
&\quad=0
\end{align*}
}
and
\begin{align*}
&\quad\partial_r\mathbb{E}[\tilde{\psi}_2(W,\theta_0,\eta_0 + r(\eta-\eta_0)]\big|_{r=0}\\
&=\mathbb{E}\left[\partial_r\tilde{\psi}_2,(W,\theta_0,\eta_0 + r(\eta-\eta_0))\big|_{r=0}\right]\\
&=\mathbb{E}\left[(\mu(D,X,1)-\mu_0(D,X,1))\right]-\mathbb{E}\left[(\mu(D,X,0)-\mu_0(D,X,0))\right]\\
&\quad-\mathbb{E}\left[\frac{Z(\mu(D,X,1)-\mu_0(D,X,1))}{p_0(D,X)}\right]\\
&\quad+\mathbb{E}\left[\frac{(1-Z)(\mu(D,X,0)-\mu_0(D,X,0))}{1-p_0(D,X)}\right]\\
&\quad-\mathbb{E}\left[\frac{Z(Y-\mu_0(1,D,X))(p(D,X)-p_0(D,X))}{p_0(D,X)^2}\right]\\
&\quad-\mathbb{E}\left[\frac{(1-Z)(Y-\mu_0(0,D,X))(p(D,X)-p_0(D,X))}{(1-p_0(D,X))^2}\right]=0
\end{align*}
since $\mathbb{E}[Z|D,X]=p_0(D,X)$, $\mathbb{E}[Z(Y-\mu_0(D,X,1))|D,X]=0$ and $\mathbb{E}[(1-Z)(Y-\mu_0(D,X,0))|D,X]=0$.

\section{Extensions to multivalued IVs}\label{app:extensionmulti}

In this appendix, we adapt the score function in eq. \eqref{score3} to multivalued IVs $Z$, which is also a key contribution of this paper. In the case of a continuous $Z$, this requires discretizing its values in some parts of the score function. To this end, let $l=1,\dots,L$ be a partition of its support $\mathcal{Z}$ with $\cup_{l}Z_l=\mathcal{Z}$. For a discrete IV, $Z_l=z_l$, $l=1,\dots,L$, would be any value $Z$ can take with probability $\mathbb{P}(Z=z_l)>c$, with $c>0$. For a continuous IV, such a partition may be generated based on the quantile function (e.g. percentiles) of $Z$. Let $1(Z\in Z_l)$ denote the indicator function, which is one if $Z$ falls into the partition $Z_l$ and else is zero, and $p_l(D,X)=\mathbb{P}(Z\in Z_l|D,X)$ denote the corresponding IV propensity score.
% \begin{align}
% Z&=p(D,X)+\nu, \quad\mathbb{E}[\nu|D,X]=0,\label{DR_model_2}
% \end{align}
Then, testing with a multivalued $Z$ can be based on the following score function:
% \begin{align}\label{score4}
% &\quad \psi(W,\theta,\eta)\\
% &=\sum_{l=1}^L1(Z\in Z_l)(\mu(D,X,Z)-m(D,X))^2\notag\\&
% +\sum_{l=1}^L2(\mu(D,X,Z)-m(D,X))\left(\frac{(Y-\mu(D,X,Z)) 1(Z\in Z_l)}{p_l(D,X)}
% -\frac{(Y-m(D,X))1(Z\notin Z_l)}{1-p_l(D,X)}\right)\notag\\&
% +\sum_{l=1}^L1(Z\in Z_l)(\mu(D,X,Z)-m(D,X))\notag\\&
% +\sum_{l=1}^L\frac{(Y-\mu(D,X,Z)) 1(Z\in Z_l)}{p_l(D,X)}
% -\frac{(Y-m(D,X)) 1(Z\notin Z_l)}{1-p_l(D,X)}-\theta\notag.
% \end{align}
% \textcolor{red}{New score:
% \begin{align*}
% \psi(W,\theta,\eta)
% &=(\mu(D,X,Z)-m(D,X))^2\notag\\&
% +2(\mu(D,X,Z)-m(D,X))\left(Y-\mu(D,X,Z)\right)\notag\\&
% -2(\mu(D,X,Z)-m(D,X))\left(Y-m(D,X)\right)\notag\\&
% +\mu(D,X,Z)-m(D,X)\notag\\&
% +(Y-\mu(D,X,Z))-(Y-m(D,X))-\theta\notag.
% \end{align*}
% }
\begin{align}\label{score4}
&\quad \psi(W,\theta,\eta)\\
&=\sum_{l=1}^L(\mu(D,X,Z\in Z_l)-\mu(D,X,Z\notin Z_l))^2\notag\\&
+\sum_{l=1}^L2(\mu(D,X,Z\in Z_l)-\mu(D,X,Z\notin Z_l))\notag\\
&\left(\frac{(Y-\mu(D,X,Z\in Z_l)) 1(Z\in Z_l)}{p_l(D,X)}
-\frac{(Y-\mu(D,X,Z\notin Z_l))1(Z\notin Z_l)}{1-p_l(D,X)}\right)\notag\\&
+\sum_{l=1}^L(\mu(D,X,Z\in Z_l)-\mu(D,X,Z\notin Z_l))\notag\\&
+\sum_{l=1}^L\left(\frac{(Y-\mu(D,X,Z\in Z_l)) 1(Z\in Z_l)}{p_l(D,X)}
-\frac{(Y-\mu(D,X,Z\notin Z_l)) 1(Z\notin Z_l)}{1-p_l(D,X)}\right)\notag.
\end{align}
This score has a variance that is bounded away from zero, is zero in expectation under the null hypothesis, $\theta_0=0$, and is Neyman-orthogonal, as formally shown in Appendix \ref{proofNeyman2}. We may construct cross-fitted estimators of $\theta_0$ based on the score function \eqref{score4}. It is worth noting that the corresponding target parameter in \eqref{score4} is given by
\begin{align*}
    \mathbb{E}\left[\sum_{l=1}^L[(\mu_0(D,X,Z\in Z_l)-\mu_0(D,X,Z\notin Z_l))^2+(\mu_0(D,X,Z\in Z_l)-\mu_0(D,X,Z\notin Z_l))]\right],
\end{align*}
which tests the null hypothesis given in eq. \eqref{nullhyp1} for binary and discrete IVs, and approximates eq. \eqref{nullhyp1} in the case of continuous IVs if the bins defining $Z_l$ become small.

\section{Proof of moment condition and Neyman orthogonality of $\psi$}\label{proofNeyman2}
Equation \eqref{score4} suggests the following score function for testing when $Z$ is multivalued discrete or continuous:  
% \begin{align}\label{score4a}
% &\quad \psi(W,\theta,\eta)\\
% &=\sum_{l=1}^L1(Z\in Z_l)(\mu(D,X,Z)-m(D,X))^2\notag\\&
% +\textcolor{red}{1/L}\sum_{l=1}^L2(\mu(D,X,Z)-m(D,X))\left(\frac{(Y-\mu(D,X,Z)) 1(Z\in Z_l)}{p_l(D,X)}
% -\frac{(Y-m(D,X))1(Z\notin Z_l)}{1-p_l(D,X)}\right)\notag\\&
% +\sum_{l=1}^L1(Z\in Z_l)(\mu(D,X,Z)-m(D,X))\notag\\&
% +\textcolor{red}{1/L}\sum_{l=1}^L\frac{(Y-\mu(D,X,Z)) 1(Z\in Z_l)}{p_l(D,X)}
% -\frac{(Y-m(D,X)) 1(Z\notin Z_l)}{1-p_l(D,X)}-\theta\notag\\
% &:=\psi_1(W,\theta,\eta)+\psi_2(W,\theta,\eta)-\theta\notag,
% \end{align}
\begin{align}
&\quad \psi(W,\theta,\eta)\\
&=\sum_{l=1}^L(\mu(D,X,Z\in Z_l)-\mu(D,X,Z\notin Z_l))^2\notag\\&
+\sum_{l=1}^L2(\mu(D,X,Z\in Z_l)-\mu(D,X,Z\notin Z_l))\notag\\
&\left(\frac{(Y-\mu(D,X,Z\in Z_l)) 1(Z\in Z_l)}{p_l(D,X)}
-\frac{(Y-\mu(D,X,Z\notin Z_l))1(Z\notin Z_l)}{1-p_l(D,X)}\right)\notag\\&
+\sum_{l=1}^L(\mu(D,X,Z\in Z_l)-\mu(D,X,Z\notin Z_l))\notag\\&
+\sum_{l=1}^L\frac{(Y-\mu(D,X,Z\in Z_l)) 1(Z\in Z_l)}{p_l(D,X)}
-\frac{(Y-\mu(D,X,Z\notin Z_l)) 1(Z\notin Z_l)}{1-p_l(D,X)}-\theta\notag\\
&:=\psi_1(W,\theta,\eta)+\psi_2(W,\theta,\eta)-\theta\notag.
\end{align}
% where $Z_l$, $l=1,\dots,L$ is a suitable partition of $\mathcal{Z}$ with $\cup_{l}Z_l=\mathcal{Z}$. First, we show that the moment condition holds. By the law of total expectation, we note
% \begin{align*}
% \mathbb{E}\left[\sum_{l=1}^L1(Z\in Z_l)(\mu_0(D,X,Z)-m_0(D,X))^2\right]&=\sum_{l=1}^L\mathbb{P}(Z\in Z_l)\mathbb{E}\left[(\mu_0(D,X,Z)-m_0(D,X))^2|Z\in Z_l\right]\\
% &=\mathbb{E}\left[(\mu_0(D,X,Z)-m_0(D,X))^2\right]
% \end{align*}
% and $\mathbb{E}[\sum_{l=1}^L1(Z\in Z_l)(\mu_0(D,X,Z)-m_0(D,X))]=\mathbb{E}[\mu_0(D,X,Z)-m_0(D,X)]$. Hence,
% $$\mathbb{E}\left[\sum_{l=1}^L1(Z\in Z_l)(\mu_0(D,X,Z)-m_0(D,X))^2\right]+\mathbb{E}\left[\sum_{l=1}^L1(Z\in Z_l)(\mu_0(D,X,Z)-m_0(D,X))\right]-\theta_0=0.$$
First, we show that the moment condition holds. By definition, we have
\begin{align*}
\sum_{l=1}^L(\mu(D,X,Z\in Z_l)-\mu(D,X,Z\notin Z_l))^2+\sum_{l=1}^L(\mu(D,X,Z\in Z_l)-\mu(D,X,Z\notin Z_l))-\theta=0.
\end{align*}
Analogous to the proof in Appendix \ref{proofNeyman}, we have
\begin{align*}
\sum_{l=1}^L\mathbb{E}\left[\frac{(Y-\mu_0(D,X,Z\in Z_l)) 1(Z\in Z_l)}{p_l(D,X)}
-\frac{(Y-\mu_0(D,X,Z\notin Z_l)) 1(Z\notin Z_l)}{1-p_l(D,X)}\right]
=0.
\end{align*}
% because
% \begin{align*}
% &\quad\sum_{l=1}^L\mathbb{E}\left[\frac{(Y-\mu_0(D,X,Z)) 1(Z\in Z_l)}{p_l(D,X)}\right]\\
% &=\sum_{l=1}^L\mathbb{E}\left[\mathbb{E}\left[\frac{(Y-\mu_0(D,X,Z)) 1(Z\in Z_l)}{p_l(D,X)}\bigg|X,D,Z\right]\right]\\
% &=\sum_{l=1}^L\mathbb{E}\left[\frac{1(Z\in Z_l)}{p_l(D,X)}\mathbb{E}\left[Y-\mu_0(D,X,Z)\bigg|X,D,Z\right]\right]=0
% \end{align*}
% since $\mathbb{E}[\varepsilon|X,D,Z]=0$ and
% \begin{align*}
% &\quad\sum_{l=1}^L\mathbb{E}\left[\frac{(Y-m_0(D,X)) 1(Z\notin Z_l)}{1-p_l(D,X)}\right]\\
% &=\sum_{l=1}^L\mathbb{E}\left[\mathbb{E}\left[\frac{(Y-m_0(D,X)) 1(Z\notin Z_l)}{1-p_l(D,X)}\bigg|X,D,Z\right]\right]\\
% &=\sum_{l=1}^L\mathbb{E}\left[\frac{1(Z\notin Z_l) }{1-p_l(D,X)}\mathbb{E}\left[Y-m_0(D,X)\bigg|X,D,Z\right]\right]= 0.
% \end{align*}
% \textcolor{red}{Do we have a problem here?? Please check:
% $$
% \mathbb{E}\left[m_0(D,X)\bigg|X,D,Z\right]=\mathbb{E}\left[\mathbb{E}[Y|X,D]\bigg|X,D,Z\right]=\mathbb{E}[Y|X,D,Z]=\mu_0(D,X,Z).
% $$
% }
Hence, we have to show that
\begin{align*}
&\mathbb{E}\Bigg[\sum_{l=1}^L(\mu_0(D,X,Z\in Z_l)-\mu_0(D,X,Z\notin Z_l))\\
&\quad \left(\frac{(Y-\mu_0(D,X,Z\in Z_l)) 1(Z\in Z_l)}{p_l(D,X)}-\frac{(Y-\mu_0(D,X,Z\notin Z_l))1(Z\notin Z_l)}{1-p_l(D,X)}\right)\Bigg]=0,
\end{align*}
which holds by the same argument. 
Next, we show that Neyman orthogonality holds. First, note that 
\begin{align*}
&\quad\partial_r\mathbb{E}[\psi_2(W,\theta_0,\eta_0 + r(\eta-\eta_0)]\big|_{r=0}\\
&=\mathbb{E}\left[\partial_r\psi_2,(W,\theta_0,\eta_0 + r(\eta-\eta_0))\big|_{r=0}\right]\\
&=\sum_{l=1}^L\Bigg(\mathbb{E}\left[(\mu(D,X,Z\in Z_l)-\mu_0(D,X,Z\in Z_l))\right]-\mathbb{E}\left[(\mu(D,X,Z\notin Z_l)-\mu_0(D,X,Z\notin Z_l))\right]\\
&\quad-\mathbb{E}\left[\frac{1(Z\in Z_l)(\mu(D,X,Z\in Z_l)-\mu_0(D,X,Z\in Z_l))}{p_l(D,X)}\right]\\
&\quad+\mathbb{E}\left[\frac{1(Z\notin Z_l)(\mu(D,X,Z\notin Z_l)-\mu_0(D,X,Z\notin Z_l))}{1-p_l(D,X)}\right]\\
&\quad-\mathbb{E}\left[\frac{1(Z\in Z_l)(Y-\mu_0(D,X,Z\in Z_l))(p(D,X)-p_l(D,X))}{p_l(D,X)^2}\right]\\
&\quad-\mathbb{E}\left[\frac{1(Z\notin Z_l)(Y-\mu_0(D,X,Z\notin Z_l))(p(D,X)-p_l(D,X))}{(1-p_l(D,X))^2}\right]\Bigg)=0
\end{align*}
since
% \begin{align*}
% &\quad\sum_{l=1}^L\Bigg(\mathbb{E}\left[1(Z\in Z_l)(\mu(D,X,Z\in Z_l)-\mu_0(D,X,Z\in Z_l))\right]-\mathbb{E}\left[1(Z\in Z_l)(m(D,X)-m_0(D,X))\right]\Bigg)\\
% &=\mathbb{E}\left[(\mu(D,X,Z)-\mu_0(D,X,Z))\right]-\mathbb{E}\left[(m(D,X)-m_0(D,X))\right]
% \end{align*}
% by law of total expectations,
\begin{align*}
&\quad\sum_{l=1}^L\mathbb{E}\left[\frac{1(Z\in Z_l)(Y-\mu_0(D,X,Z\in Z_l))(p(D,X)-p_l(D,X))}{p_l(D,X)^2}\right]\\
&=\sum_{l=1}^L\mathbb{E}\left[\frac{1(Z\notin Z_l)(Y-\mu_0(D,X,Z\notin Z_l))(p(D,X)-p_l(D,X))}{(1-p_l(D,X))^2}\right]=0,
\end{align*}
by the same arguments used before,
\begin{align*}
&\quad\sum_{l=1}^L\mathbb{E}\left[\frac{1(Z\notin Z_l)(\mu(D,X,Z\notin Z_l)-\mu_0(D,X,Z\notin Z_l))}{1-p_l(D,X)}\right]\\
&=\sum_{l=1}^L\mathbb{E}\left[\mathbb{E}\left[\frac{1(Z\notin Z_l)(\mu(D,X,Z\notin Z_l)-\mu_0(D,X,Z\notin Z_l))}{1-p_l(D,X)}\Bigg|D,X\right]\right]\\
&=\sum_{l=1}^L\mathbb{E}\left[(\mu(D,X,Z\notin Z_l)-\mu_0(D,X,Z\notin Z_l))\mathbb{E}\left[\frac{1(Z\notin Z_l)}{1-p_l(D,X)}\Bigg|D,X\right]\right]\\
% &=\sum_{l=1}^L\mathbb{E}\left[\mathbb{E}\left[(\mu(D,X,Z\notin Z_l)-\mu_0(D,X,Z\notin Z_l))|Z\notin Z_l\right]\right]\\
% &=\sum_{l=1}^L\mathbb{P}(Z\notin Z_l)\mathbb{E}\left[\mathbb{E}\left[\frac{(\mu(D,X,Z\notin Z_l)-\mu_0(D,X,Z\notin Z_l))}{1-p_l(D,X)}|Z\notin Z_l\right]\Bigg|X,D\right]\\
% &= \sum_{l=1}^L\mathbb{E}\left[(\mu(D,X,Z\notin Z_l)-\mu_0(D,X,Z\notin Z_l))\mathbb{E}\left[\frac{1(Z\notin Z_l)}{1-p_l(D,X)}\bigg|D,X\right]\right]\\
&=\sum_{l=1}^L\mathbb{E}\left[(\mu(D,X,Z\notin Z_l)-\mu_0(D,X,Z\notin Z_l))\right]
\end{align*}
by iterated expectation and
\begin{align*}
\quad&\sum_{l=1}^L\mathbb{E}\left[\frac{1(Z\in Z_l)(\mu(D,X,Z\in Z_l)-\mu_0(D,X,Z\in Z_l))}{p_l(D,X)}\right]\\
&=\sum_{l=1}^L\mathbb{E}\left[(\mu(D,X,Z\in Z_l)-\mu_0(D,X,Z\in Z_l))\right].
\end{align*}
% since $\mathbb{E}[p_l(D,X)|D,X]=\mathbb{P}(Z\in Z_l)$, $\mathbb{E}[1(Z\in Z_l)(Y-\mu_0(D,X,1))|D,X]=0$ and $\mathbb{E}[1(Z\notin Z_l)(Y-\mu_0(D,X,0))|D,X]=0$.
% \textcolor{red}{Here I am not sure what to do since it is not sufficient to condition on $X$ and $D$. Any ideas?}
% \textcolor{blue}{Idea: Redefine the score:
% \begin{align*}
% &\quad\sum_{l=1}^L\mathbb{E}\left[1(Z\in Z_l)(\mu(D,X,Z)-\mu_0(D,X,Z))\right]\\
% % &=\mathbb{E}\left[\mu(D,X,Z)-\mu_0(D,X,Z)\mathbb{E}\left[\frac{1(Z\in Z_l)}{p_l(D,X)}\bigg|D,X,Z\right]\right]\\
% &=\sum_{l=1}^L\mathbb{P}(Z\in Z_l)\mathbb{E}\left[(\mu(D,X,Z)-\mu_0(D,X,Z))\bigg|Z\in Z_l\right]\\
% % &=\mathbb{E}\left[\mathbb{E}\left[(\mu(D,X,Z)-\mu_0(D,X,Z))/p_l(D,X)\bigg|Z\in Z_l,D,X\right]\right]\\
% &=\mathbb{E}\left[(\mu(D,X,Z)-\mu_0(D,X,Z))\right]
% % &=\mathbb{E}\left[\mathbb{E}\left[(\mu(D,X,Z)-\mu_0(D,X,Z))\bigg|Z\in Z_l\right]\mathbb{E}\left[\frac{\mathbb{P}(Z\in Z_l)}{p_l(D,X)}\bigg|D,X\right]\right].
% \end{align*}
% }
Further, we have
{\scriptsize
\allowdisplaybreaks{
\begin{align*}
&\quad \partial_r \mathbb{E}[\psi_1(W, \theta_0, \eta_0 + r(\eta - \eta_0))]\big|_{r=0} \\
&= \mathbb{E}\left[\partial_r \psi_1(W, \theta_0, \eta_0 + r(\eta - \eta_0))\big|_{r=0}\right] \\
&= 2 \sum_{l=1}^L \mathbb{E}\Big[(\mu_0(D, X, Z \in Z_l) - \mu_0(D, X, Z \notin Z_l))\\
&\quad\cdot
\left((\mu(D, X, Z \in Z_l) - \mu_0(D, X, Z \in Z_l)) - (\mu(D, X, Z \notin Z_l) - \mu_0(D, X, Z \notin Z_l))\right)\Big] \\
&\quad - 2 \sum_{l=1}^L \mathbb{E}\left[(\mu_0(D, X, Z \in Z_l) - \mu_0(D, X, Z \notin Z_l)) 
\cdot \frac{1(Z \in Z_l)(\mu(D, X, Z \in Z_l) - \mu_0(D, X, Z \in Z_l))}{p_l(D, X)}\right] \\
&\quad + 2 \sum_{l=1}^L \mathbb{E}\left[(\mu_0(D, X, Z \in Z_l) - \mu_0(D, X, Z \notin Z_l)) 
\cdot \frac{1(Z \notin Z_l)(\mu(D, X, Z \notin Z_l) - \mu_0(D, X, Z \notin Z_l))}{1 - p_l(D, X)}\right] \\
&\quad - 2 \sum_{l=1}^L \mathbb{E}\left[(\mu_0(D, X, Z \in Z_l) - \mu_0(D, X, Z \notin Z_l)) 
\cdot \frac{1(Z \in Z_l)(Y - \mu_0(D, X, Z \in Z_l))(p(D, X) - p_l(D, X))}{p_l(D, X)^2}\right] \\
&\quad - 2 \sum_{l=1}^L \mathbb{E}\left[(\mu_0(D, X, Z \in Z_l) - \mu_0(D, X, Z \notin Z_l)) 
\cdot \frac{1(Z \notin Z_l)(Y - \mu_0(D, X, Z \notin Z_l))(p(D, X) - p_l(D, X))}{(1 - p_l(D, X))^2}\right] \\
&\quad + 2 \sum_{l=1}^L \mathbb{E}\left[(\mu_0(D, X, Z \in Z_l) - \mu_0(D, X, Z \notin Z_l)) 
\cdot \frac{1(Z \notin Z_l)(Y - \mu_0(D, X, Z \notin Z_l))(p(D, X) - p_l(D, X))}{(1 - p_l(D, X))^2}\right] \\
&\quad + 2 \sum_{l=1}^L \mathbb{E}\bigg[((\mu(D, X, Z \in Z_l) - \mu_0(D, X, Z \in Z_l)) - (\mu(D, X, Z \notin Z_l) - \mu_0(D, X, Z \notin Z_l))) \\
&\quad \left(\frac{(Y - \mu_0(D, X, Z \in Z_l)) \cdot 1(Z \in Z_l)}{p_l(D, X)} - \frac{(Y - \mu_0(D, X, Z \notin Z_l)) \cdot 1(Z \notin Z_l)}{1 - p_l(D, X)}\right)\bigg] \\
&= 2 \sum_{l=1}^L \mathbb{E}\Big[(\mu_0(D, X, Z \in Z_l) - \mu_0(D, X, Z \notin Z_l))\\
&\quad\cdot \left((\mu(D, X, Z \in Z_l) - \mu_0(D, X, Z \in Z_l)) - (\mu(D, X, Z \notin Z_l) - \mu_0(D, X, Z \notin Z_l))\right)\Big] \\
&\quad - 2 \sum_{l=1}^L \mathbb{E}\left[(\mu_0(D, X, Z \in Z_l) - \mu_0(D, X, Z \notin Z_l)) 
\cdot (\mu(D, X, Z \in Z_l) - \mu_0(D, X, Z \in Z_l))\right] \\
&\quad + 2 \sum_{l=1}^L \mathbb{E}\left[(\mu_0(D, X, Z \in Z_l) - \mu_0(D, X, Z \notin Z_l)) 
\cdot (\mu(D, X, Z \notin Z_l) - \mu_0(D, X, Z \notin Z_l))\right] \\
&= 0.
\end{align*}
}
}

\section{Proof of Theorem \ref{theorem1}}\label{proofth1}
To prove Theorem \ref{theorem1}, we apply Theorem 3.1 in \citet{doubleML}. Hence, we only need to show Assumption 3.1 and 3.2 in \citet{doubleML}. All bounds in the proof hold uniformly over $\mathbb{P}\in\mathcal{P}$ but we omit this qualifier for brevity. We use $C$ to denote a strictly positive constant that is independent of $n$ and $\mathbb{P}\in\mathcal{P}$. The value of $C$ may change at each appearance. In Appendix \ref{proofNeyman2} we have already shown that the moment condition and Neyman orthogonality is satisfied. Next, we note that the score in equation \eqref{score4} is linear, i.e.
$$\psi(W,\theta,\eta)=\psi^a(W,\eta)\theta+\psi^b(W,\eta),$$
with $\psi^a(W,\eta)=-1,$ which is in line with Assumption 3.1 in \cite{doubleML}. Next, we demonstrate the satisfaction of Assumption 3.2 to complete the proof. We define the following nuisance realization set $\mathcal{T}_n$ as the set of all P-square-integrable functions $\eta$ such that
\begin{align*}
\|\eta_{0}-\eta\|_{P,2q}&\le C,\\
\|\eta_{0}-\eta\|_{P,4}&\le \delta_N,\\
\|\eta_{0}-\eta\|_{P,2}&\le\delta_N^{1/2}N^{-1/4},
\end{align*}
for $\delta_N=o(1)$ and a constant $q>2$. Note that Assumption 3.2 (a)-(c) hold by construction of the set $\mathcal{T}_N$ and Assumption \ref{assnorm} due to same arguments as in \citet{huberkueck2022}. Assumption 3.2 (d) in \cite{doubleML} holds since
\begin{align*}
&\quad\mathbb{E}[\psi(W,\theta_0,\eta_0)^2]\\
&\ge \mathbb{E}\left[\left(\sum_{l=1}^L\frac{(Y-\mu_0(D,X,Z\in Z_l)) 1(Z\in Z_l)}{p_l(D,X)}
-\frac{(Y-\mu_0(D,X,Z\notin Z_l)) 1(Z\notin Z_l)}{1-p_l(D,X)}\right)^2\right]\\
&>0
\end{align*}
by Assumption $\ref{assnorm}$ and since $p_l(D,X)>c>0$ by construction for all $l=1,\dots,L$. This completes the proof.

\section{Proof of Theorem \ref{prop:all}}
\begin{proof}
Since $\lim\limits_{n\rightarrow\infty} P(\hat{S}=	\mathcal{S}) = 1$ we just have to show that
$$\lim\limits_{n\rightarrow\infty} P(\hat{V}\subseteq	\mathcal{V}) = 1. $$
If $\hat{\mathcal{P}}_{pass}=\emptyset$, our statement holds trivially. Hence, consider a variable $j: \mathcal{P}_j\in\hat{\mathcal{P}}_{pass}$, i.e, $Z_j\in\hat{S}$ and $Z_j\in\hat{\mathcal{V}}$. Therefore, it holds
$$|\sqrt{n}\hat{\sigma}^{-1}_j\hat{\theta}_j| < c_\alpha$$
by definition of $\hat{\mathcal{V}}$. By Corollary \ref{coro:test}, we have that
   $$\lim\limits_{n\rightarrow\infty}\, P(|\sqrt{n}\hat{\sigma}^{-1}_j\hat{\theta}_j|> c_\alpha) = 1 $$
if $j\notin\mathcal{V}$. This shows that $j\in\hat{\mathcal{V}}$ implies  $j\in{\mathcal{V}}$ which concludes the proof.
\end{proof}

\section{Proof of Theorem \ref{power}}
\begin{proof}
Consider any variable $Z_j\in \mathcal{P}^*$. We know that 
$E[Y|D,X]=E[Y|D,X,Z_j]$ and hence by Theorem \ref{theorem1},
$$P(\underbrace{|\sqrt{n}\hat{\sigma}_j^{-1}\hat{\theta}_j|>c_\alpha}_{:=A_j})= \alpha$$
if $n\rightarrow\infty$. 
% It follows that $P(j\in\hat{\mathcal{P}}_{pass})= 1-\alpha$ for each $j\in\mathcal{P}^*$. 
We can conclude that
\begin{align*}
\lim\limits_{n\rightarrow\infty}P(\hat{\mathcal{P}}_{pass}\neq\emptyset)&=1-\lim\limits_{n\rightarrow\infty}P(\hat{\mathcal{P}}_{pass}=\emptyset)\\
&=1-\lim\limits_{n\rightarrow\infty}P(\cap_{j=1,\dots,M}A_j)\ge 1-\alpha
\end{align*}
with $M:=|\mathcal{P}^*|>0$. Understanding the stochastic dependence structure of the events $A_j$, $j=1,\dots,M$, could obviously lead to sharper bounds. 
\end{proof}
% correlation structure of our different test statistics $|\sqrt{n}\hat{\sigma}_j^{-1}\hat{\theta}_j|$,  in Theorem \ref{theorem1} could help to minimize the volume of the joint confidence intervals and hence to improve the power of our identification test in \eqref{identification}.
\clearpage

\begin{samepage}
\section{Pseudo-code}

\begin{algorithm}
\begin{footnotesize}
\caption{Testing procedure for selecting partition $\mathcal{P}^*$}\label{algo:procedure}
\rule{\linewidth}{0.4pt} % Custom thickness
\SetAlgoLined
\DontPrintSemicolon
\KwIn{$Y$ (Outcome), $D$ (Treatment), $Q$ (Potential controls and instruments)}
\BlankLine
\textbf{Select strong instruments $\hat{S}$}\\
\For{$j \in Q$}{
    Predict $D$ by $Q_j$ conditional on $Q\setminus Q_j$ using ML or regression: $\hat{F}_j\gets \hat{t}_j^2$\
    
    \If{$\hat{F}_j > C_{\tilde\alpha}$ with $\tilde{\alpha}=0.1/\log(n)$}{$j \in \hat{S}$}
}
\BlankLine
\textbf{Select instruments passing conditional independence test $\hat{\mathcal{V}}$}\ \\
\For{$j \in \hat{S}$}{
    $\hat{\mu}(D,X,Z) \gets$ and $\hat{p}(D,X) \gets$ ML estimates\\
    $\hat{\theta}_j \gets$ constructed via score in eq. \eqref{score3} (binary case) or \eqref{score4} (continuous case)\\
    $\hat{\sigma}_j \gets$ constructed as described in Theorem \ref{theorem1}\\
    $\hat{t}_{ind,j}= \frac{\hat{\theta}_j}{\hat{\sigma}_j} \gets$ test $H_0: \theta = 0$\\
    \If{$\hat{t}_{ind,j} < c_{\alpha}$}{
        $j \in \hat{\mathcal{V}}$\;
    }
}
\BlankLine
\textbf{Select final partition $\hat{\mathcal{P}}$}: $\hat{\mathcal{P}}_{pass} \gets \left\{\mathcal{P}_j: j \in \left(\hat{S} \cap \hat{\mathcal{V}}\right) \right\}$\\
\If{$\hat{\mathcal{P}}_{pass} = \emptyset$}{
    end and report that $H_0$ is not rejected (no identification)\;
}\ElseIf{$|\hat{\mathcal{P}}_{pass}| = 1$}{
    $\hat{\mathcal{P}} \gets \hat{\mathcal{P}}_{pass}$
}\ElseIf{$|\hat{\mathcal{P}}_{pass}| > 1$}{
    $\texttt{p}(\mathcal{P}_{j}) \gets \texttt{p-value}(\hat{t}_{ind,j})$ for $j \in \hat{\mathcal{V}}$\\\
    $\hat{\mathcal{P}}_{pmax} = \underset{\mathcal{P}_j \in \hat{\mathcal{P}}_{pass}}{\arg\max} \ p(\mathcal{P}_j)$\
}
\BlankLine
\textbf{If identification is confirmed ($\hat{\mathcal{P}}_{pass} \neq \emptyset$), estimate the treatment effect}\\
Regress: $Y \leftarrow D$ and $X=Q_{[j]} \in \hat{\mathcal{P}}_{pmax}$\;

\rule{\linewidth}{0.4pt} % Custom thickness
\end{footnotesize}
\end{algorithm}
\end{samepage}

\newpage
\section{Computational cost}\label{app:computationalcost}

In low-dimensional settings, we recommend using the FSHT procedure, which is computationally attractive. We agree that in high-dimensional settings the computational costs can be substantially larger. In the paper (see p.~8), we recommend using double machine learning (DML) to estimate the first-stage effect for each of the $p$ candidates when $p>n$.
A (weakly) smaller number of computations is performed for the conditional mean-independence test in the second stage because the number of candidate variables that pass the first-stage screening is typically smaller. Denoting by $s\le p$ the number of strong IV candidates, the overall computational order can be summarized as
\[
O\!\left(p\cdot C_{\mathrm{DML}} + s\cdot C_{\mathrm{test}}\right),
\]
where $C_{\mathrm{DML}}$ denotes the cost of one DML estimation for a single candidate variable (including $K$-fold cross-fitting and nuisance estimation), and $C_{\mathrm{test}}$ denotes the cost of one second-stage test evaluation.
In our simulations and application we use, for instance, random forests as learners for nuisance estimation. The computational cost of a random forest with $T$ trees typically scales (up to constants and implementation details) on the order of $O(T \cdot m \cdot n \log n)$, where $m$ is the number of features considered per split.
When $p$ is very large, it is crucial to use efficient methods for running the repeated first-stage estimations across the $p$ candidates. In this case, we recommend $L_2$-boosting as a computationally attractive alternative for nuisance estimation and screening, as discussed in Section~3.3 of \citet{KUECK2023714}. Using $L_2$-boosting primarily reduces the per-fit cost $C_{\mathrm{DML}}$ (and hence the overall runtime), while the outer loop over candidates remains. Importantly, the computations are embarrassingly parallel across candidate variables, so runtime can be substantially reduced via parallel computing.

\clearpage

\section{Simulation Details}\label{app:simulation}

%As a word of caution for empirical applications, our simulations also show that $n$ might need to be quite substantial for our testing approach to perform adequately. 
%, as the finite sample performance in our simulations is not convincing in moderate samples of just several thousand observations.

This section provides a simulation study to investigate the finite sample behavior of our testing approach based on the following data generating process (DGP):
\begin{eqnarray*}
	%Y_2 &=& D_1 + D_2 + (\gamma D_1)\cdot X_0 + X_0'\beta_{X_0}  + X_1'\beta_{X_1}  + U ,\\
	Y &=& D + X'\beta + \gamma Z + W  + U,\\
	D &=& I\{X'\beta + Z + \delta W  + V >0\},\\
	X, Z &\sim& Bernoulli(\pi), \\
W&\sim& N(0,1),  U\sim N(0,1), V\sim N(0,1),
\end{eqnarray*}
with $X,Z,W,U,V$ being independent of each other. Outcome $Y$ is a linear function of $D$ (whose treatment effect is one), covariates $X$ (for $\beta\neq 0$), the unobservables $W$ and $U$, and the supposed instrument $Z$ if the coefficient $\gamma\neq 0$. The binary treatment $D$ is a function of $X$ and the unobservable $V$, as well as $W$ if coefficient $\delta\neq 0$. While the supposed IV $Z$ is binary, the unobserved terms $U,V,W$ are random, standard normally distributed variables that are independent of  each other, of $Z$, and of $X$. The joint set of covariates and IVs $Q$ is created as follows. We first draw a matrix of multivariate normal variables $\tilde{Q}\sim N(\mathbf{0}, \bm{\Sigma})$ where $\bm{\Sigma}$ is a covariance matrix with $Cov(\tilde{Q}_j,\tilde{Q}_k) = 0.5^{|j-k|}$. We then compute probabilities $\pi_i = \frac{1}{1+exp(-2\cdot\tilde{Q}_i)}$ and draw $Q_i\sim Bernoulli(\pi_i)$. We generate 10 variables in the set $Q=(X,Z)$, with the last one being the (single) IV $Z$. Concerning the first nine elements in $Q$, which are considered as covariates $X$, $\beta$  gauges their effects on $Y$ and $D$, respectively, and thus, the magnitude of confounding due to observables. The $j$th element in the coefficient vector $\beta$ is set to $0.8/j$ for $j=1,...,4$, implying a linear decay of covariate importance in terms of confounding for the first 4 covariates. Moreover, for elements $j=5,...,9$, the coefficients are equal to zero, implying that covariates 5 to 9 are random noise variables that are neither confounders, nor IVs.

%Covariates $X$ are created as follows. We first draw a matrix of multivariate normal variables $\tilde{X}\sim N(\mathbf{0}, \bm{\Sigma})$ where $\bm{\Sigma}$ is a covariance matrix with $Cov(X_j,X_k) = 0.5^{|j-k|}$. We then compute probabilities $\pi_i = \frac{1}{1+exp(-\tilde{X}_i)}$ and draw $X_i\sim Bernoulli(\pi_i)$. $\beta$  gauges the effects of the covariates on $Y$ and $D$, respectively, and thus, the magnitude of confounding due to observables. The $j$th element in the coefficient vector $\beta$ is set to $0.8/j$ for $j=1,...,4$, implying a linear decay of covariate importance in terms of confounding for the first 4 covariates. Furthermore, for elements $j=5,...,9$, the coefficients are equal to zero, implying that covariates 5 to 9 are random noise variables. Having nine covariates in $X$ and one IV $Z$ implies that there are in total 10 variables in the set of observed variables $Q=(X,Z)$, with the first 4 being confounders, and the last one being the IV. 

In a second simulation design, we also consider continuous variables $Q$, which differ from the preceding DGP in that $X, Z \sim Uniform(-0.5, 0.5)$, i.e. the set of observed covariates $X$ and the IV $Z$ are now drawn from the multivariate uniform distribution, $X, Z \sim Uniform(-0.5, 0.5)$. All other properties of this DGP are identical to the previous scenario with binary $Z$ and $X$ variables. %Specifically, the outcome $Y$ is characterized by the same linear function of $D, X, W$, and $U$ (and the supposed IV $Z$ in the case where $\gamma\neq 0$) as before, and we again consider 10 variables in the set $Q=(X,Z)$, with the last one being the (single) IV $Z$, elements $j=1,...,4$ being observed confounders, and elements $j=5,...,9$ being random noise variables that are neither confounders nor IVs.

We investigate the performance of our testing approach in $200$ simulations with sample sizes of $n=1000$, $4000$, and $16000$. For estimating the (conditional) first stage effect of elements in $Q$ on $D$ to select sufficiently strong IVs, we use DML for partially linear models \citep{doubleML}. DML can be applied to both discrete and continuous elements in $Q$ and is implemented in the \textit{DoubleML} package for statistical software \textsf{R}. The so-called nuisance parameters in the first stage are estimated using the LASSO with cross-fitting, setting the number of folds to five (K = 5). Specifically, the classification LASSO is used for discrete elements in $Q$, modelling probabilities for treatment assignment while the LASSO is used for continuous variables. For testing the conditional independence of any selected candidate IV in the second stage, we consider a cross-fitted estimator $\hat\theta$ of eq. \eqref{eq:test1} with five folds ($K=5$), which is based on the doubly robust score functions in eqs. \eqref{score3} and \eqref{score4} for binary and continuous candidate IVs, respectively.
%when using the doubly robust score function in eq. \eqref{score3}. 
In the case of a continuous IV, its support is partitioned based on the quartiles of its distribution, such that $L=4$.  For nuisance parameter estimation in the second stage, LASSO with default parameters as implemented in the \textit{glmnet} package \citep{Friedmanetal2010} for the statistical software \textsf{R} is applied. We consider several statistics from our simulations: the estimated violation, $\hat\theta$, when using variable $Z$ as the IV, the standard deviation of the estimated violation (std), and the average of the standard error of the estimated violation across all simulation samples (mean se). %to see whether it approximates the standard deviation well. 
These statistics are useful for judging the performance of the test conditional on selecting $Z$ as the IV, which is the correct choice if $\gamma=0$, while there is no valid IV if $\gamma \neq 0$ and no identification when controlling for covariates if $\delta \neq 0$. 

Since the choice of the IV is not fixed a priori but is part of the estimation process, we also consider three further statistics. The first one is the empirical selection rate of $Z$ (sel.Z), which indicates the frequency with which $Z$ is selected as the best candidate IV \textit{and also} has a \texttt{p}-value greater than 30\% when testing conditional independence based on the estimate $\hat\theta$.\footnote{We suggest the threshold of 30\% based on examination of the distribution of \texttt{p}-values for both true IVs and confounders, see Figures \ref{fig:pvals_oracle} and \ref{fig:pvals_conf}. This stricter threshold guards against the possible choice of a confounder as an IV, although a more traditional threshold of 10\% could also be used.} In other words, this corresponds to the proportion of simulations in which our method simultaneously selects $Z$ as the IV and at the same time keeps the null hypothesis of conditional independence. In scenarios where the null hypothesis holds ($\delta=\gamma=0$), this selection rate should approach one as $n$ increases. The second statistic is the analogous selection rate of the noise variables $X_5$ to $X_9$ (sel.noconf). It indicates the share of simulations in which one of the non-confounders (that are not IVs either) is selected as the best candidate IV and at the same time yields a \texttt{p}-value greater than 30\% when testing conditional independence when using this noise variable as the IV.
As $n$ increases, this selection rate should always approach zero, because asymptotically, noise variables do not have a first stage effect on the treatment conditional on other elements in $Q$. Therefore, noise variables should not be selected as candidate IVs for the conditional independence test. Finally, we report an equivalent selection rate for the confounders $X_1$ to $X_4$ (sel.conf). It corresponds to the share of simulations in which one of the confounders is selected as the best candidate IV and at the same time yields a \texttt{p}-value greater than 30\% when testing conditional independence when using this confounder as IV. As $n$ increases, this selection rate should always approach zero, because confounders are never valid IVs.
%Finally, we also report the  frequency of having any candidate IVs passing both the first stage selection and entailing a p-value larger than 10\% in the conditional independence test at all (sel.any), which is based on the union of $Z$ and confounders  $X_1$ to $X_4$. Formally, this corresponds to the share of simulations with nonempty intersections in $S_{pass}$ and $V_{pass}$. Asymptotically, this selection rate should go to one if the null hypothesis holds, as $Z$ should be selected with probability $1$, and to zero if the null is violated, as neither $Z$ or confounders $X_1$ to $X_4$ satisfy conditional independence in this case. 
\begin{table}[h]
%\spacingset{1}
\begin{center}
	\caption{Simulations: Single Instrument}
	\label{tab:sim}
    \begin{footnotesize}
	\begin{tabular}{c|ccc |c c c }
 \hline\hline
  \multicolumn{7}{l}{\textbf{PANEL A: Binary Instrument}} \\
  \hline
 N & sel.Z & sel.noconf & sel.conf & $\hat{\theta}$ & std & mean se\\  %& sel.any\\
  \hline
    &\multicolumn{6}{c}{Assumptions \ref{ass1} and \ref{ass2} hold ($\delta=0$, $\gamma=0$)} \\
   \hline
1000 & 32\% & 2\% & 59\% & 0.002 & 0.010 & 0.008 \\ 
  4000 & 45\% & 1\% & 42\% & 0.000 & 0.002 & 0.002 \\ 
  16000 & 68\% & 2\% & 4\% & 0.000 & 0.000 & 0.000 \\ 
   \hline
   &\multicolumn{6}{c}{A\ \ref{ass1} violated, A\ \ref{ass2} holds ($\delta=2$, $\gamma=0$)} \\
   \hline
  1000 & 50\% & 3\% & 32\% & 0.012 & 0.012 & 0.016 \\ 
  4000 & 15\% & 4\% & 44\% & 0.012 & 0.006 & 0.009 \\ 
  16000 & 0\% & 4\% & 20\% & 0.013 & 0.003 & 0.005 \\
 \hline
  &\multicolumn{6}{c}{A\ \ref{ass1} holds, A\ \ref{ass2} violated ($\delta=0$, $\gamma=0.5$)} \\
   \hline
  1000 & 30\% & 2\% & 57\% & 0.045 & 0.040 & 0.050 \\ 
  4000 & 3\% & 2\% & 68\% & 0.053 & 0.019 & 0.027 \\ 
  16000 & 0\% & 3\% & 8\% & 0.054 & 0.010 & 0.014 \\ 
   \hline
     \multicolumn{7}{l}{\textbf{PANEL B: Continuous Instrument}} \\
     \hline
 N & sel.Z & sel.noconf & sel.conf & $\hat{\theta}$ & std & mean se\\  %& sel.any\\
  \hline
    &\multicolumn{6}{c}{Assumptions \ref{ass1} and \ref{ass2} hold ($\delta=0$, $\gamma=0$)} \\
   \hline
1000 & 55\% & 2\% & 19\% & -0.007 & 0.029 & 0.031 \\ 
  4000 & 65\% & 4\% & 8\% & -0.001 & 0.015 & 0.016 \\ 
  16000 & 64\% & 1\% & 0\% & 0.000 & 0.007 & 0.009 \\  
   \hline
   &\multicolumn{6}{c}{A\ \ref{ass1} violated, A\ \ref{ass2} holds ($\delta=2$, $\gamma=0$)} \\
   \hline
  1000 & 16\% & 6\% & 6\% & -0.018 & 0.031 & 0.017 \\ 
  4000 & 8\% & 2\% & 1\% & -0.015 & 0.018 & 0.007 \\ 
  16000 & 0\% & 2\% & 0\% & -0.030 & 0.014 & 0.003 \\ 
 \hline
  &\multicolumn{6}{c}{A\ \ref{ass1} holds, A\ \ref{ass2} violated ($\delta=0$, $\gamma=0.5$)} \\
   \hline
  1000 & 0\% & 4\% & 18\% & -0.490 & 0.192 & 0.049 \\ 
  4000 & 0\% & 6\% & 8\% & -0.552 & 0.103 & 0.026 \\ 
  16000 & 0\% & 1\% & 0\% & -0.554 & 0.055 & 0.013 \\ 
   \hline
\end{tabular}
     \end{footnotesize}
\end{center}
\par
{\scriptsize Notes: columns `$\hat{\theta}$', `std', and  `mean se' provide the average estimate of $\theta$ (the violation of conditional independence) conditional on using $Z$ as IV, its standard deviation, and the average of the standard errors across all samples, respectively. `sel.Z' gives the share of simulations in which $Z$ is selected as best candidate IV and simultaneously yields a \texttt{p}-value $> 0.30$ when testing conditional independence based on the estimate of $\hat{\theta}$. `sel.noconf' is an equivalent measure for covariates $X_5$ to $X_9$, corresponding to the share of selecting a non-confounder as best candidate IV and having a \texttt{p}-value $>$ 30\% when testing conditional independence. `sel.conf' is an equivalent measure for covariates $X_1$ to $X_4$, corresponding to the share of selecting an observable confounder the best candidate IV and having a \texttt{p}-value $>$ 30\% when testing conditional independence. %`sel.noconf' is an equivalent measure for confounders $X_1$ to $X_4$. It corresponds to the frequency with which one of the confounders is selected as the best candidate IV and also entails a p-value of more than 10\% when testing conditional independence. %'sel.any' is the frequency of having any candidate IVs passing both the first stage selection and entailing a p-value larger than 10\% in the conditional independence test at all. This corresponds to the share of simulations with nonempty intersections in $S_{pass}$ and $V_{pass}$. 
}
\end{table}
%\spacingset{1.8}

Table \ref{tab:sim} reports the simulation results for binary variables in $Q$ in Panel A, and for continuous elements in $Q$ in Panel B, respectively. %In the binary case, $Z$ is correlated with $X_j$ according to the covariance matrix $\bm{\Sigma}$. 
In both panels, the top rows provide the results for $\delta=\gamma=0$, implying that both Assumptions \ref{ass1} and \ref{ass2} are satisfied and conditional independence holds. As $n$ increases, the test is more likely to simultaneously select the true IV $Z$ and keep the null hypothesis of conditional independence. Specifically, the selection rate of $Z$ (sel.Z) increases from 32\% with $n=1000$ observations to 68\% with $n=16000$ observations for the binary IV, and from 55\% with $n=1000$ to 64\% with $n=16000$ for the continuous IV. This improvement is mirrored by a reduction in the rate at which a confounder variable ($X_1...X_4$) is selected as the IV and passes the test (sel.conf). Under $n=16000$, it only occurs that any confounder is selected as the best IV and passes the conditional independence test 4\% of the time for the binary and never for the continuous scenario. However, we observe that the noise variables ($X_5...X_9$) sometimes appear to step in, as their selection rate (sel.noconf) transiently increases in response to the decreasing selection rate of confounders. Additionally, we observe that conditional on having selected the true IV, the estimated violation $\hat\theta$ %($\hat{\theta}$)
is close to zero for any sample size and never statistically significant at conventional significance levels. Figure \ref{fig:sim_estimates} depicts the distribution of the estimated $\hat\theta$ and visually confirms this finding in the first column. The top row of density plots show the binary $\hat\theta$ and the bottom row shows the continuous estimates. As sample size increases, the distribution moves closer to zero, as expected.  

The intermediate rows of Table \ref{tab:sim} and middle column of Figure \ref{fig:sim_estimates} present the results for a violation of SOO when setting $\delta=2$, $\gamma=0$. As $n$ increases, the selection rates of $Z$ (sel.Z) and the confounder variables (sel.conf) both go towards zero, as expected. The rate of finding a nonempty set $\hat{\mathcal{P}}_{pass}$, which is equal to the sum of the selection rates sel.Z, sel.noconf and sel.conf, also goes towards zero, finding no identification as expected. Furthermore, the estimated violation conditional on using $Z$ as an IV ($\hat{\theta}$) is statistically different from zero under the larger sample sizes $n=4000,16000$. We note that in settings where A \ref{ass1} is violated, the selection rate of a confounder (sel.conf) in the binary scenario is still slightly high, around 20\% in the largest sample for the binary IV, athough this rate is 0\% for the continuous IV.\footnote{In further experiments (not reported here), strengthening the degree of confounding by increasing the value of $\gamma$ leads to a substantial reduction in the rate at which a confounder is selected.}  The lower rows of Table \ref{tab:sim} and final column of Figure \ref{fig:sim_estimates} provide the performance measures under a violation of A \ref{ass2}, IV validity, when considering $\delta=0$, $\gamma=0.5$. Again and as expected, both the selection rates of $Z$ (sel.Z) and the confounder variables (sel.conf) tend to zero as $n$ increases. At the same time, the estimated violation conditional on using $Z$ as an IV ($\hat{\theta}$) is large in absolute terms and highly statistically significant across all sample sizes, and $\hat{\mathcal{P}}_{pass}$ is empty 89\% of the time for binary and 99\% of the time for continuous IVs for $n = 16000$.

%\subsection{Multiple Instruments}

We now consider a modification of the previous DGP that entails the existence of multiple IVs. Specifically, the true set of IVs $(Z_1...Z_3)$ now comprises the last three elements in $Q$, i.e. variables 8 to 10, which are constructed equivalently to $Z$ in the previous setup. We investigate the performance when A \ref{ass2} is violated for all of the IVs. This implies that the set of confounders remains the same as before  ($X_1...X_4$), while the set of noise variables now only consists of three elements ($X_5...X_7$). Results for the multiple IV simulations are shown in Table \ref{tab:sim2}, and are similar to those shown in the single IV results. In the case of multiple binary IVs and no violations of the assumptions, the selection rate of $Z$ (sel.Z) is up to 96\% when $n=16000$. As in the single IV scenario, under a violation of A \ref{ass1}, SOO, the selection rates of the true IVs (sel.Z) and confounding variables (sel.conf) go to zero. For multiple continuous IVs (Panel B), results are consistent with the single IV setting, with the exception of the slightly lower increase of the IV selection rate in the larger sample size compared to the binary version (67\% in the case of multiple IVs, compared to 64\% under the single IV scenario). However, in all settings of the multiple continuous IV case, the rate of selecting a confounder (sel.conf) goes to zero and the estimate magnitudes and $\hat{\mathcal{P}}_{pass}$ behave as expected, moving towards zero and the empty set under no violation of the assumptions and away from zero otherwise. 

%\spacingset{1}
\begin{table}[h]
\begin{center}
	\caption{Simulations: Multiple Instruments}
	\label{tab:sim2}
    \begin{footnotesize}
	\begin{tabular}{c|ccc |c c c  }
 \hline\hline
  \multicolumn{7}{l}{\textbf{PANEL A: Binary Instruments}} \\
  \hline
 N & sel.Z & sel.noconf & sel.conf & $\hat{\theta}$ & std & mean se\\  %& sel.any\\
  \hline
    &\multicolumn{5}{c}{Assumptions \ref{ass1} and \ref{ass2} hold ($\delta=0$, $\gamma=0$)} \\
   \hline
1000 & 80\% & 1\% & 18\% & 0.002 & 0.008 & 0.007 \\ 
  4000 & 84\% & 2\% & 12\% & 0.000 & 0.002 & 0.002 \\ 
  16000 & 96\% & 0\% & 0\% & 0.000 & 0.000 & 0.000 \\ 
   \hline
   &\multicolumn{5}{c}{A\ \ref{ass1} violated, A\ \ref{ass2} holds ($\delta=2$, $\gamma=0$)} \\
   \hline
 1000 & 86\% & 1\% & 11\% & 0.010 & 0.012 & 0.014 \\ 
  4000 & 40\% & 2\% & 40\% & 0.010 & 0.005 & 0.007 \\ 
  16000 & 0\% & 3\% & 16\% & 0.011 & 0.003 & 0.004 \\
   \hline
     &\multicolumn{5}{c}{A\ \ref{ass1} holds, A\ \ref{ass2} violated ($\delta=0$, $\gamma=0.5$ $\forall Z$)} \\
   \hline
  1000 & 70\% & 0\% & 21\% & 0.053 & 0.035 & 0.050 \\ 
  4000 & 15\% & 2\% & 48\% & 0.058 & 0.018 & 0.026 \\ 
  16000 & 0\% & 1\% & 8\% & 0.058 & 0.009 & 0.013 \\
  \hline
\multicolumn{7}{l}{\textbf{PANEL B: Continuous Instruments}} \\
\hline
 N & sel.Z & sel.noconf & sel.conf & $\hat{\theta}$ & std & mean se\\  %& sel.any\\
  \hline
    &\multicolumn{5}{c}{Assumptions \ref{ass1} and \ref{ass2} hold ($\delta=0$, $\gamma=0$)} \\
   \hline
1000 & 62\% & 1\% & 6\% & -0.010 & 0.035 & 0.031 \\ 
  4000 & 64\% & 1\% & 4\% & -0.003 & 0.017 & 0.016 \\ 
  16000 & 67\% & 2\% & 0\% & -0.001 & 0.010 & 0.008 \\ 
   \hline
   &\multicolumn{5}{c}{A\ \ref{ass1} violated, A\ \ref{ass2} holds ($\delta=2$, $\gamma=0$)} \\
   \hline
  1000 & 38\% & 2\% & 2\% & -0.027 & 0.036 & 0.017 \\ 
  4000 & 14\% & 0\% & 1\% & -0.025 & 0.021 & 0.007 \\ 
  16000 & 0\% & 0\% & 0\% & -0.040 & 0.007 & 0.003 \\
 \hline
     &\multicolumn{5}{c}{A\ \ref{ass1} holds, A\ \ref{ass2} violated ($\delta=0$, $\gamma=0.5$ $\forall Z$)} \\
   \hline
  1000 & 20\% & 2\% & 10\% & -0.116 & 0.090 & 0.037 \\ 
  4000 & 0\% & 2\% & 6\% & -0.142 & 0.064 & 0.019 \\ 
  16000 & 0\% & 2\% & 0\% & -0.140 & 0.017 & 0.010 \\ 
   \hline
\end{tabular}
\end{footnotesize}
\end{center}
\par
{\scriptsize Notes: columns `$\hat{\theta}$', `std', and  `mean se' provide the average estimate of $\theta$ (the violation of conditional independence) conditional on using $Z_3$ (the last element in $Q$) as IV, its standard deviation, and the average of the standard errors across all samples, respectively. `sel.Z' gives the share of simulations in which one of the IVs $Z_1$ to $Z_3$ is selected as best candidate IV and simultaneously yields a \texttt{p}-value $> 0.30$ when testing conditional independence based on the estimate of $\hat{\theta}$. `sel.noconf' is an equivalent measure for covariates $X_5$ to $X_7$, corresponding to the share of selecting a non-confounder as best candidate IV and having a \texttt{p}-value $>$ 0.3\% when testing conditional independence. `sel.conf' is an equivalent measure for covariates $X_1$ to $X_4$, corresponding to the share of selecting an observable confounder the best candidate IV and having a \texttt{p}-value $>$ 0.3 when testing conditional independence. %`sel.noconf' is an equivalent measure for confounders $X_1$ to $X_4$. It corresponds to the frequency with which one of the confounders is selected as the best candidate instrument and also entails a p-value of more than 10\% when testing conditional independence. %'sel.any' is the frequency of having any candidate instruments passing both the first stage selection and entailing a p-value larger than 10\% in the conditional independence test at all. This corresponds to the share of simulations with nonempty intersections in $S_{pass}$ and $V_{pass}$. 
}
\end{table}
%\spacingset{1.8}

\setcounter{table}{0}
\renewcommand{\thetable}{A\arabic{table}}

\setcounter{figure}{0}
\renewcommand{\thefigure}{A\arabic{figure}}
\clearpage
\section{Figures}

\begin{figure}[h!]
\centering
\captionsetup{labelformat=empty}
\caption{Simulations: Distribution of Estimates}
\label{fig:sim_estimates}

% Top row: Binary IV
\begin{minipage}{0.33\textwidth}
    \centering
    \includegraphics[width=\linewidth, keepaspectratio]{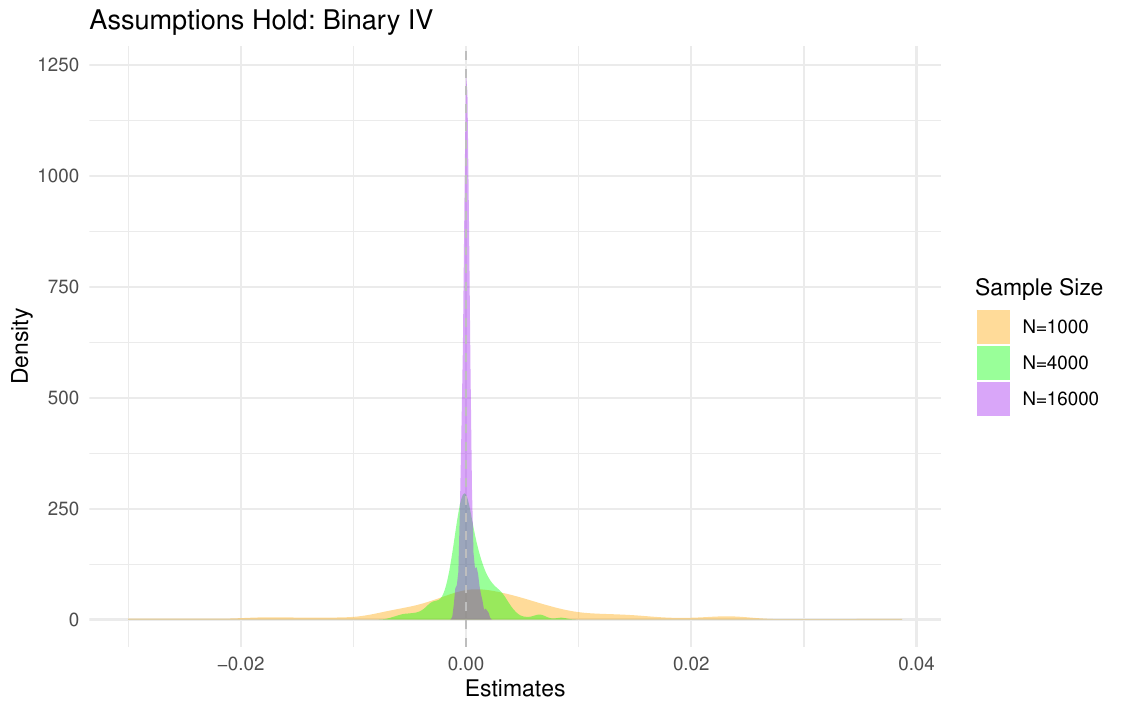}
\end{minipage}%
\hfill
\begin{minipage}{0.33\textwidth}
    \centering
    \includegraphics[width=\linewidth, keepaspectratio]{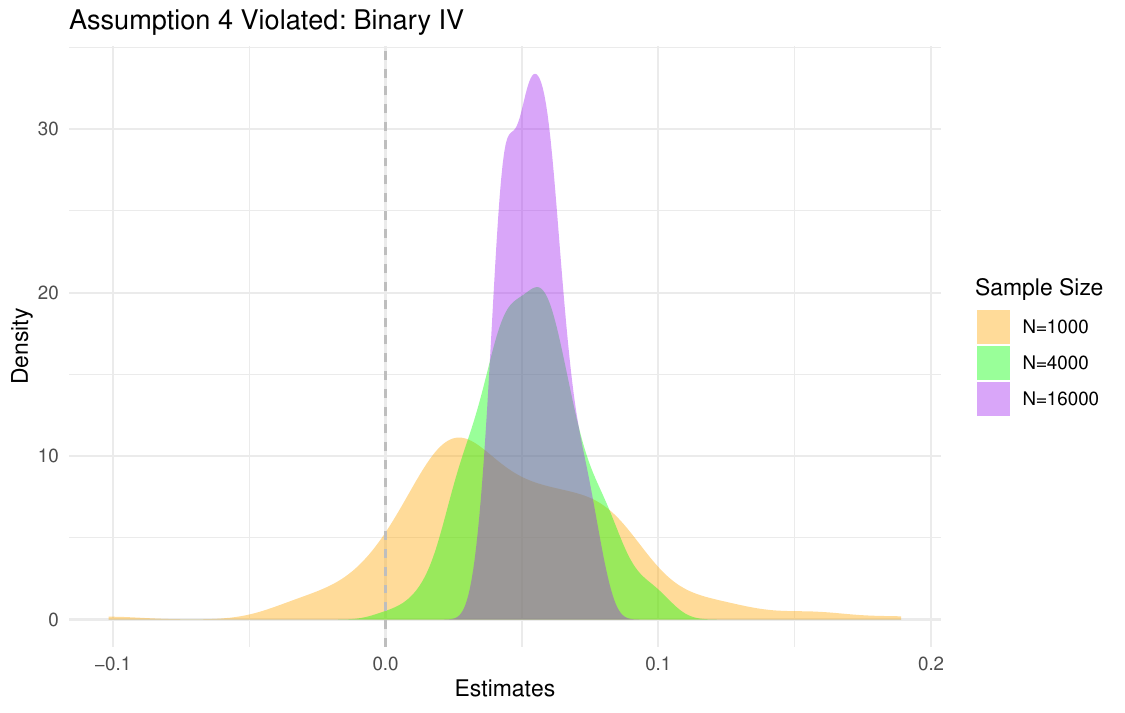}
\end{minipage}%
\hfill
\begin{minipage}{0.33\textwidth}
    \centering
    \includegraphics[width=\linewidth, keepaspectratio]{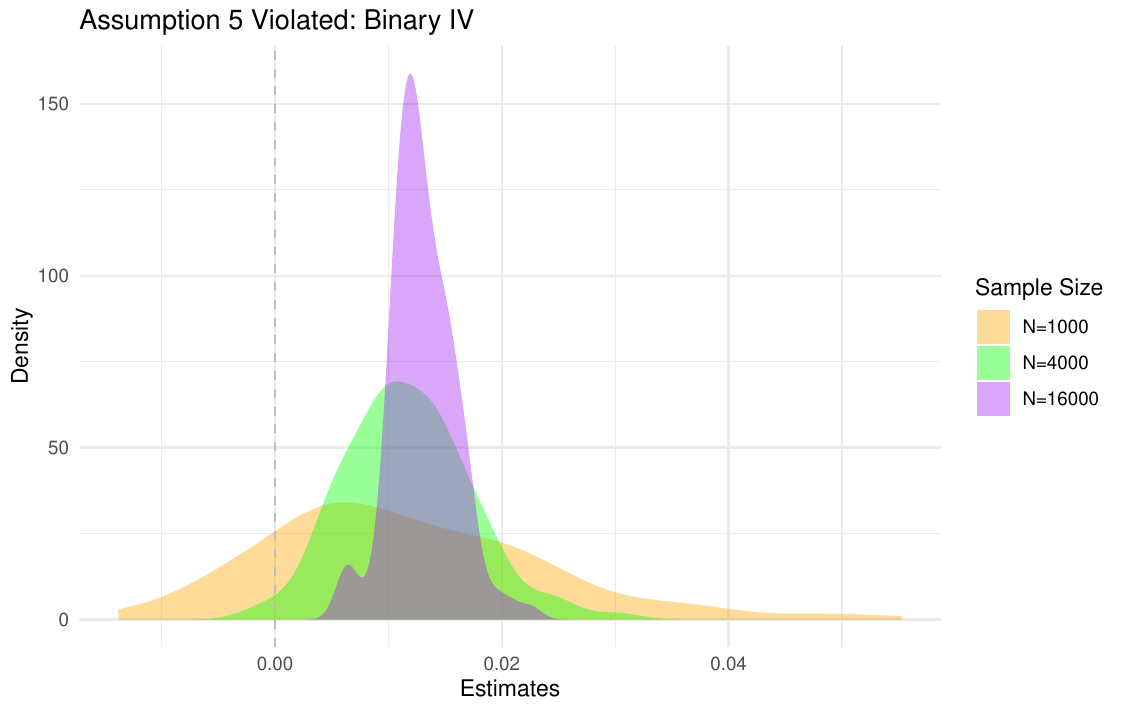}
\end{minipage}

\vspace{0.3cm} % Adjust spacing between rows

% Bottom row: Continuous IV
\begin{minipage}{0.33\textwidth}
    \centering
    \includegraphics[width=\linewidth, keepaspectratio]{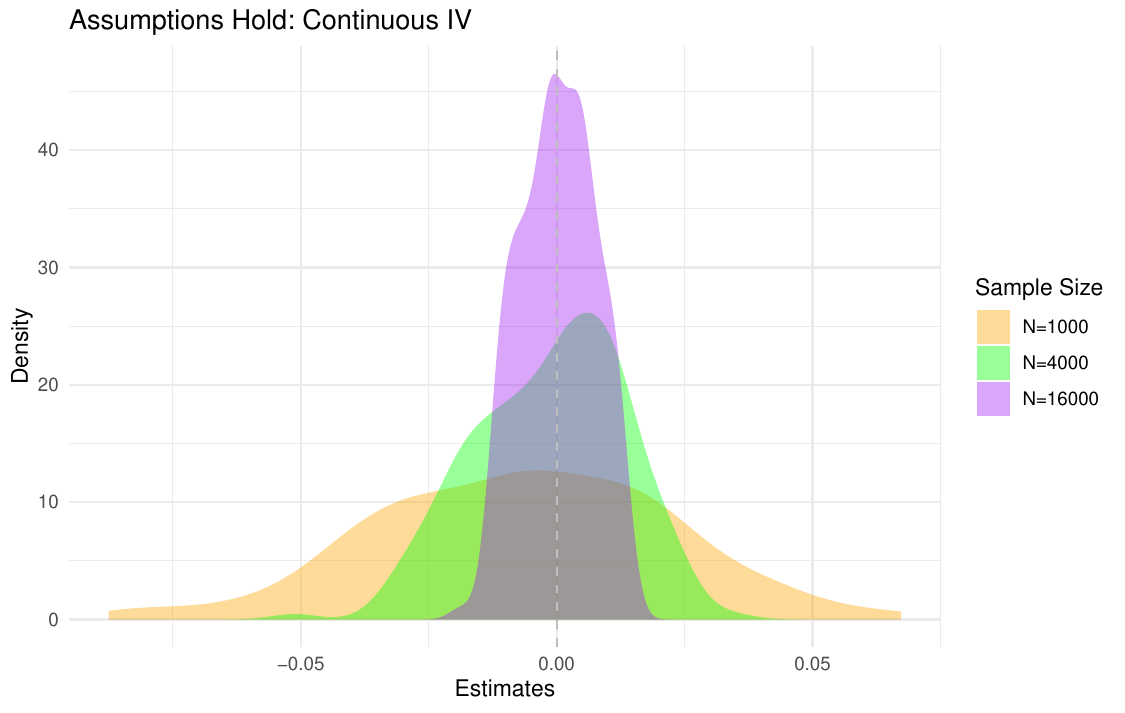}
\end{minipage}%
\hfill
\begin{minipage}{0.33\textwidth}
    \centering
    \includegraphics[width=\linewidth, keepaspectratio]{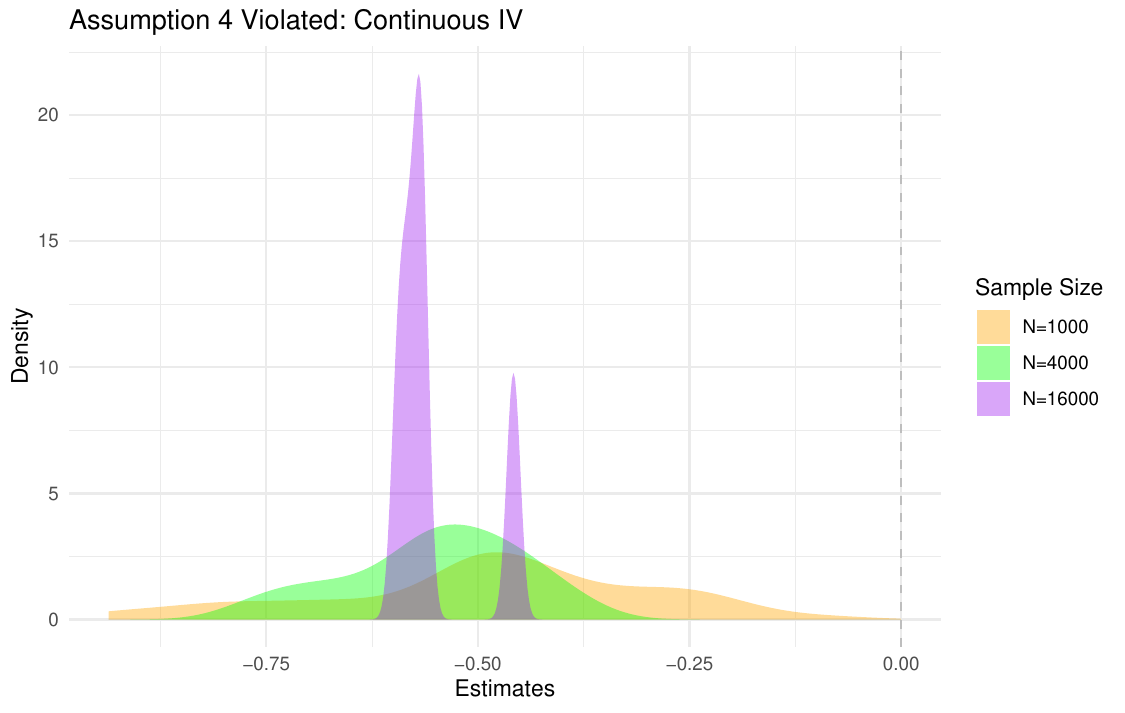}
\end{minipage}%
\hfill
\begin{minipage}{0.33\textwidth}
    \centering
    \includegraphics[width=\linewidth, keepaspectratio]{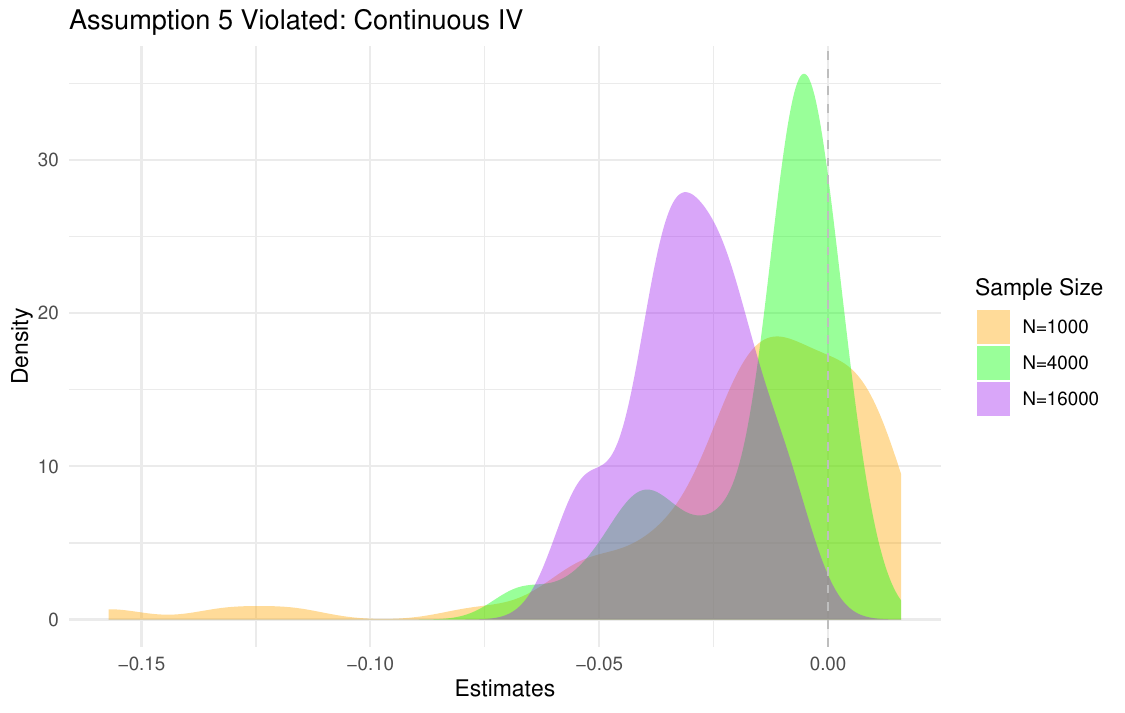}
\end{minipage}

\caption*{\footnotesize{This figure depicts density plots of the estimated violation $\hat{\theta}$ for the true IV (covariate $X_{10}$), across the simulation settings. The top and bottom rows are the binary and continuous IV cases, respectively. The first plots on the left provide results for $\delta=\gamma=0$, i.e. Assumptions \ref{ass1} and \ref{ass2} are satisfied. The middle plots depict the scenario of A \ref{ass1} being violated, and the right-most plots A \ref{ass2} is violated. Orange is the smallest sample size (N = 1000), green the medium sample size (N=4000), and purple is the largest (N=16000).}}
\label{densityplot}
\end{figure}

\begin{figure}
    \centering
    \caption{Distribution of \texttt{p}-values: Oracle IV, Binary Simulations}
    \includegraphics[width=0.85\linewidth]{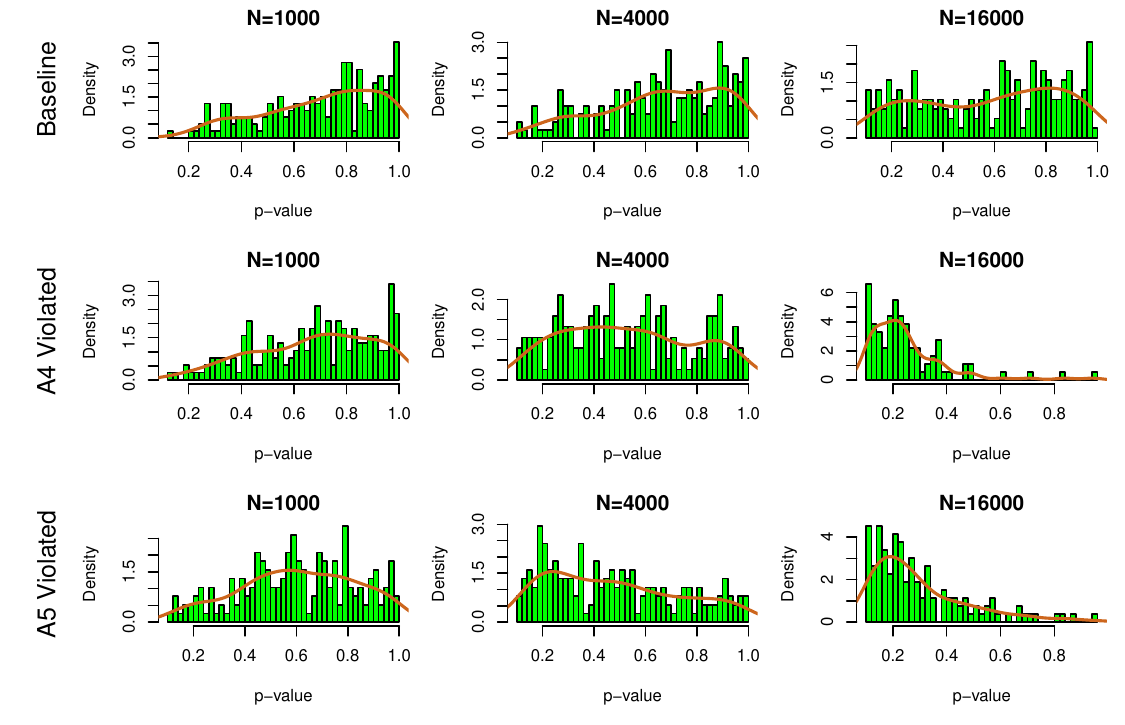}
    \label{fig:pvals_oracle}
\end{figure}

\begin{figure}
    \centering
    \caption{Distribution of \texttt{p}-values: Confounder $X_1$, Binary Simulations}
    \includegraphics[width=0.85\linewidth]{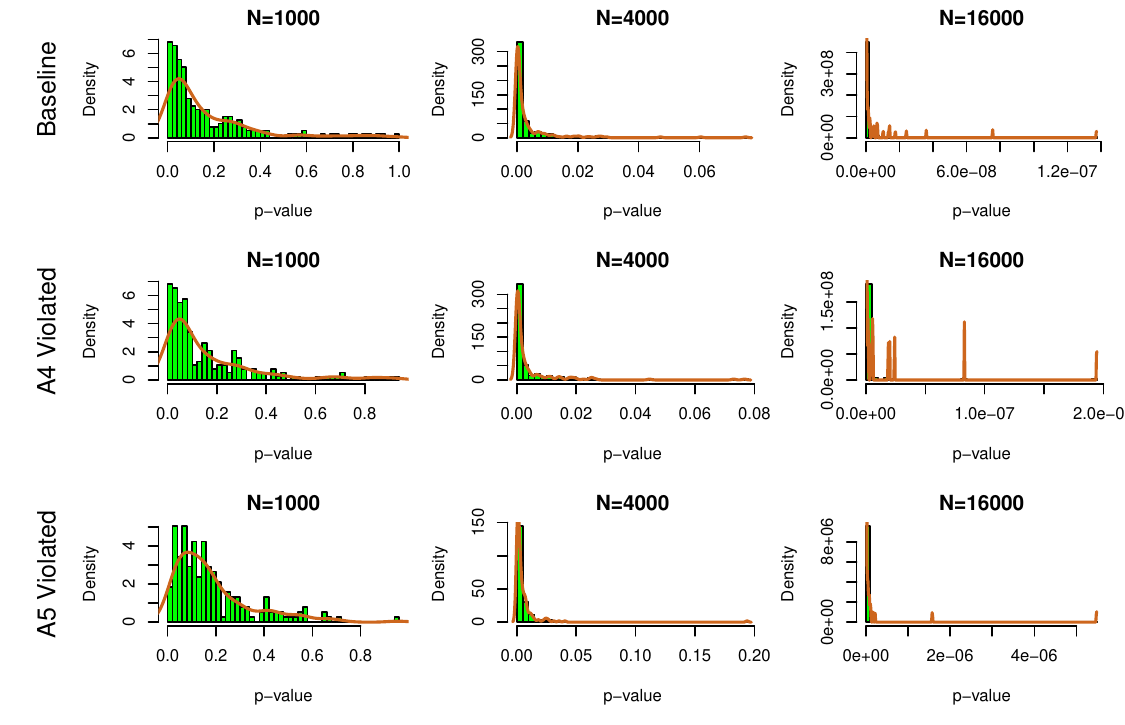}
    \label{fig:pvals_conf}
\end{figure}

\newpage
\section{Tables}

\begin{table}[ht]
\centering
\caption{Application: All $\hat{S}$}
\label{app_fullests}
\begin{tabular}{rrrrr}
  \hline
  \multicolumn{5}{l}{\textbf{Panel A: Number of Prescriptions}} \\
 Candidate IV & est & se & pval & pct trimmed \\ 
  \hline
  HH Income & 0.003 & 0.038 & 0.929 & 0.386 \\ 
  Random Assignment Instrument & -0.007 & 0.010 & 0.472 & 0.001 \\ 
  Living with a partner & 0.013 & 0.009 & 0.163 & 0.335 \\ 
  Ever Surveyed & 0.013 & 0.009 & 0.143 & 0.294 \\ 
  Any RX & -0.019 & 0.011 & 0.078 & 0.391 \\ 
  Need Dental Assistance Missing & 0.060 & 0.032 & 0.063 & 0.905 \\ 
  Application Received & 0.024 & 0.013 & 0.061 & 0.381 \\ 
  Employer Insurance & 0.025 & 0.013 & 0.058 & 0.551 \\ 
  Application Approved & 0.040 & 0.015 & 0.007 & 0.681 \\ 
  Surveyed All Months & -0.540 & 0.176 & 0.002 & 0.771 \\ 
  Retired & 0.109 & 0.031 & 0.000 & 0.572 \\ 
  Application Pending & 0.057 & 0.016 & 0.000 & 0.703 \\ 
  All Survey End & 0.027 & 0.007 & 0.000 & 0.325 \\ 
  OHP Insurance Missing & 0.037 & 0.009 & 0.000 & 0.410 \\ 
  Any Hospitalisation & 0.057 & 0.014 & 0.000 & 0.656 \\ \\ 
   \hline
    \multicolumn{5}{l}{\textbf{Panel B: Number of Doctor Visits}} \\
Random Assignment Instrument & -0.000 & 0.013 & 0.999 & 0.002 \\ 
  Any Doctor Visit & -0.002 & 0.011 & 0.871 & 0.344 \\ 
  Application Approved & 0.004 & 0.021 & 0.867 & 0.694 \\ 
  Retired & -0.027 & 0.073 & 0.705 & 0.601 \\ 
  Application Received & 0.010 & 0.015 & 0.496 & 0.398 \\ 
  Need Dental Work Missing & 0.039 & 0.052 & 0.447 & 0.917 \\  
  Ever Surveyed & 0.011 & 0.011 & 0.287 & 0.315 \\ 
  Surveyed in All Months & -0.238 & 0.217 & 0.274 & 0.785 \\ 
  Under 19 All Insured & 0.026 & 0.016 & 0.111 & 0.475 \\ 
  Ending Survey & 0.023 & 0.012 & 0.064 & 0.343 \\ 
  Application Pending & 0.051 & 0.024 & 0.036 & 0.720 \\ 
  HH Income Category & 0.112 & 0.053 & 0.036 & 0.402 \\ 
  Any Hospitalisation & 0.064 & 0.025 & 0.012 & 0.669 \\ 
  More than 30 hours employed & 0.044 & 0.011 & 0.000 & 0.350 \\ 
  Insurance Missing Variable & 0.053 & 0.011 & 0.000 & 0.424 \\ 
  OOP Prescription Costs & -0.071 & 0.002 & 0.000 & 0.436 \\ 
  OOP Medical Costs & -0.060 & 0.001 & 0.000 & 0.429 \\ \\ 
   \hline
\end{tabular}
\label{tab:application_allcovar}
\end{table}

\end{document}